\documentclass[12pt, onecolumn,a4paper,accepted=2026-02-16]{quantumarticle}
\pdfoutput=1 
\usepackage[T1]{fontenc}
\input{Qcircuit}
\usepackage{makecell}
\usepackage{tikz}
\usetikzlibrary{positioning, calc}
\usetikzlibrary{calc}
\usetikzlibrary{arrows}
\usepackage{tikz-3dplot}
\usepackage{graphicx}
\usepackage{mdwlist}
\usepackage{color}
\usepackage{thmtools}
\usepackage{amssymb}
\usepackage{physics}
\usepackage{amsmath}
\usepackage{array}
\usepackage{doi}
\usepackage{url,hyperref}
\usepackage[capitalize,nameinlink]{cleveref}
\hypersetup{colorlinks={true},linkcolor={blue},citecolor=red}

\usepackage[numbers,sort&compress]{natbib}
\usetikzlibrary{fadings}
\usetikzlibrary{decorations.pathmorphing}
\usepackage{mathrsfs}
\usepackage{authblk}
\usepackage[mathscr]{euscript}
\usepackage{enumitem}

\DeclareMathAlphabet\mathbfcal{OMS}{cmsy}{b}{n}

\usepackage{cleveref}

\usetikzlibrary{decorations.pathreplacing}
\tikzset{snake it/.style={decorate, decoration=snake}}

\usetikzlibrary{decorations.pathreplacing,decorations.markings}

\tikzset{
    >=stealth',
    punkt/.style={
           rectangle,
           rounded corners,
           draw=black, very thick,
           text width=6.5em,
           minimum height=2em,
           text centered},
    pil/.style={
           ->,
           thick,
           shorten <=2pt,
           shorten >=2pt,},
  on each segment/.style={
    decorate,
    decoration={
      show path construction,
      moveto code={},
      lineto code={
        \path [#1]
        (\tikzinputsegmentfirst) -- (\tikzinputsegmentlast);
      },
      curveto code={
        \path [#1] (\tikzinputsegmentfirst)
        .. controls
        (\tikzinputsegmentsupporta) and (\tikzinputsegmentsupportb)
        ..
        (\tikzinputsegmentlast);
      },
      closepath code={
        \path [#1]
        (\tikzinputsegmentfirst) -- (\tikzinputsegmentlast);
      },
    },
  },
  mid arrow/.style={postaction={decorate,decoration={
        markings,
        mark=at position .5 with {\arrow[#1]{stealth'}}
      }}}
}

\usepackage{hyperref}
\usepackage{slashed}
\usepackage{float} 
\usepackage{graphicx} 
\usepackage{subcaption}

\usepackage{bbold}

\mathchardef\mhyphen="2D
\newcommand{\forr}{\mathrm{forr}}

\newcommand{\coNP}{\mathsf{coNP}}

\newcommand{\IP}{\mathsf{IP}}

\newcommand{\PP}{\mathsf{PP}}

\newcommand{\QAM}{\mathsf{QAM}}

\newcommand{\AM}{\mathsf{AM}}
\newcommand{\QIP}{\mathsf{QIP}}
\newcommand{\QPP}{\mathsf{QPP}}
\newcommand{\QNP}{\mathsf{QNP}}
\newcommand{\coQNP}{\mathsf{coQNP}}
\newcommand{\HVSZK}{\mathsf{HVSZK}}
\newcommand{\HVQSZK}{\mathsf{HVQSZK}}
\newcommand{\R}{\mathsf{R}}
\newcommand{\Q}{\mathsf{Q}}
\newcommand{\QMAM}{\mathsf{QMAM}}

\newcommand{\CDS}{\textnormal{CDS}}
\newcommand{\CDQS}{\textnormal{CDQS}}

\newcommand{\cc}{\textit{cc}}

\newtheorem{theorem}{Theorem}

\newtheorem{corollary}[theorem]{Corollary}

\newtheorem{definition}[theorem]{Definition}

\newtheorem{lemma}[theorem]{Lemma}

\newtheorem{remark}[theorem]{Remark}

\newenvironment{proof}[1][Proof]{\noindent\textbf{#1. }}{\ \rule{0.5em}{0.5em}}
\begin{document}

\title{Comparing classical and quantum conditional disclosure of secrets}

\author[1]{Uma Girish}
\email{ug2150@columbia.edu}
\orcid{}

\author[2,3]{Alex May}
\email{amay@perimeterinstitute.ca}
\orcid{0000-0002-4030-5410}

\author[1]{Leo Orshansky}
\email{lo2433@columbia.edu}
\orcid{}

\author[2]{Chris Waddell}
\email{cwaddell@perimeterinstitute.ca}
\orcid{}

\affiliation[1]{Columbia University}
\affiliation[2]{Perimeter Institute for Theoretical Physics}
\affiliation[3]{Institute for Quantum Computing, Waterloo, Ontario}

\abstract{The conditional disclosure of secrets (CDS) setting is among the most basic primitives studied in information-theoretic cryptography.  
Motivated by a connection to non-local quantum computation and position-based cryptography, CDS with quantum resources has recently been considered. 
Here, we study the differences between quantum and classical CDS, with the aims of clarifying the power of quantum resources in information-theoretic cryptography.
We establish the following results:
\begin{itemize}
    \item We prove a $\Omega(\log \R_{0,A\rightarrow B}(f)+\log \R_{0,B\rightarrow A}(f))$ lower bound on quantum CDS where $\R_{0,A\rightarrow B}(f)$ is the classical one-way communication complexity with perfect correctness. 
    \item We prove a lower bound on quantum CDS in terms of two round, public coin, two-prover interactive proofs.
    \item For perfectly correct CDS, we give a separation for a promise version of the not-equals function, showing a quantum upper bound of $O(\log n)$ and classical lower bound of $\Omega(n)$. 
    \item We give a logarithmic upper bound for quantum CDS on forrelation, while the best known classical algorithm is linear. We interpret this as preliminary evidence that classical and quantum CDS are separated even with correctness and security error allowed. 
\end{itemize}
We also give a separation for classical and quantum private simultaneous message passing for a partial function, improving on an earlier relational separation.  
Our results use novel combinations of techniques from non-local quantum computation and communication complexity.}

\maketitle

\pagebreak

\tableofcontents

\flushbottom

\section{Introduction}

The conditional disclosure of secrets (CDS) setting \cite{gertner1998protecting} is among the simplest and best studied settings in information-theoretic cryptography. 
Classically, it has applications in attribute based encryption \cite{gay2015communication}, private information retrieval \cite{gertner1998protecting}, secret sharing \cite{applebaum2020power}, and has a number of connections to communication complexity \cite{applebaum2021placing}.
Recently, CDS has begun to be studied in the quantum setting. 
This first arose because of a connection between information-theoretic cryptography, including CDS, and non-local quantum computation \cite{allerstorfer2024relating}.
Quantum CDS also later appeared in the context of quantum gravity and the AdS/CFT correspondence \cite{may2024cryptographic}. 
Some basic properties of quantum CDS were established in \cite{asadi2024conditional}, including amplification, closure under constant depth formulas, and several lower bounds from communication complexity. 

The CDS scenario involves three parties, Alice, Bob and the referee. 
Alice receives input $x\in X=\{0,1\}^n$, Bob receives input $y\in Y=\{0,1\}^n$, and the referee knows both $x$ and $y$. 
Alice additionally holds a secret $s\in S$. 
An instance of CDS is specified by a choice of Boolean function $f:X\times Y\rightarrow \{0,1\}$. 
Alice and Bob can share randomness (in the classical case) or entanglement (in the quantum case). 
From their inputs and shared correlation, Alice and Bob each produce a message which they send simultaneously to the referee. 
Their goal is for the referee to be able to recover $s$ when $f(x,y)=1$, but not learn anything about $s$ when $f(x,y)=0$. 
This is illustrated in figure \ref{fig:CDSandCDQS} and defined formally in definitions \ref{def:CDS} and \ref{def:CDQS}. 
We use ``robust CDS'' to refer to settings in which we allow non-zero soundness and correctness error, ``perfectly secure'' or ``perfectly correct CDS'' for protocols lacking the respective type of error, and ``perfect CDS'' when there is neither type of error.

\begin{figure*}
    \centering
    \begin{subfigure}{0.45\textwidth}
    \centering
    \begin{tikzpicture}[scale=0.4]
    
    \draw[thick] (-5,-5) -- (-5,-3) -- (-3,-3) -- (-3,-5) -- (-5,-5);
    
    \draw[thick] (5,-5) -- (5,-3) -- (3,-3) -- (3,-5) -- (5,-5);
    
    \draw[thick] (5,5) -- (5,3) -- (3,3) -- (3,5) -- (5,5);
    
    \draw[thick, mid arrow] (4,-3) -- (4.5,3);
    
    \draw[thick, mid arrow] (-4,-3) to [out=90,in=-90] (3.5,3);
    
    \draw[thick,dashed] (-3.5,-5.5) -- (3.5,-5.5);
    \draw[black] plot [mark=*, mark size=3] coordinates{(-3.5,-5.5)};
    \draw[black] plot [mark=*, mark size=3] coordinates{(3.5,-5.5)};
    \node[below] at (0,-5.5) {$r$};
    
    \draw[thick] (-4.5,-6) -- (-4.5,-5);
    \node[below] at (-4.5,-6) {$x,s$};
    
    \draw[thick] (4.5,-6) -- (4.5,-5);
    \node[below] at (4.5,-6) {$y$};

    \node[left] at (0,1) {$m_A$};
    \node[right] at (4.5,0) {$m_B$};
    
    \draw[thick] (4,5) -- (4,6);
    \node[above] at (4,6) {$s$ iff $f(x,y)=1$};
    
    \end{tikzpicture}
    \caption{}
    \label{fig:CDQS}
    \end{subfigure}
    \hfill
    \begin{subfigure}{0.45\textwidth}
    \centering
    \begin{tikzpicture}[scale=0.4]
    
    \draw[thick] (-5,-5) -- (-5,-3) -- (-3,-3) -- (-3,-5) -- (-5,-5);
    
    \draw[thick] (5,-5) -- (5,-3) -- (3,-3) -- (3,-5) -- (5,-5);
    
    \draw[thick] (5,5) -- (5,3) -- (3,3) -- (3,5) -- (5,5);
    
    \draw[thick, mid arrow] (4,-3) -- (4.5,3);
    
    \draw[thick, mid arrow] (-4,-3) to [out=90,in=-90] (3.5,3);
    
    \draw[thick] (-3.5,-5) to [out=-90,in=-90] (3.5,-5);
    \draw[black] plot [mark=*, mark size=3] coordinates{(0,-7.05)};

    \node[left] at (0,1) {$M_A$};
    \node[right] at (4.5,0) {$M_B$};
    
    \draw[thick] (-4.5,-6) -- (-4.5,-5);
    \node[below] at (-4.5,-6) {$x, Q$};
    
    \draw[thick] (4.5,-6) -- (4.5,-5);
    \node[below] at (4.5,-6) {$y$};
    
    \draw[thick] (4,5) -- (4,6);
    \node[above] at (4,6) {Q iff $f(x,y)=1$};
    
    \end{tikzpicture}
    \caption{}
    \label{fig:PSQMintro}
    \end{subfigure}
    \caption{(a) A classical CDS protocol. Alice, on the lower left, holds input $x\in \{0,1\}^n$ and a secret $s$ from alphabet $S$. Bob, on the lower right, holds input $y\in \{0,1\}^n$. Alice and Bob can share a random string $r$. The referee, top right, holds $x$ and $y$. Alice sends a message $m_A(x,s,r)$ to the referee; Bob sends a message $m_B(y,r)$. The referee should learn $s$ iff $f(x,y)=1$ for some agreed on choice of Boolean function $f$. (b) A quantum CDS protocol. The secret can be a quantum system $Q$ or classical string $s$ (the two cases are equivalent, as noted in \cite{allerstorfer2024relating}). Alice and Bob can share an entangled quantum state, and send quantum messages to the referee. The referee should be able to recover the secret iff $f(x,y)=1$. Figure reproduced from \cite{asadi2024conditional}.} 
    \label{fig:CDSandCDQS}
\end{figure*}
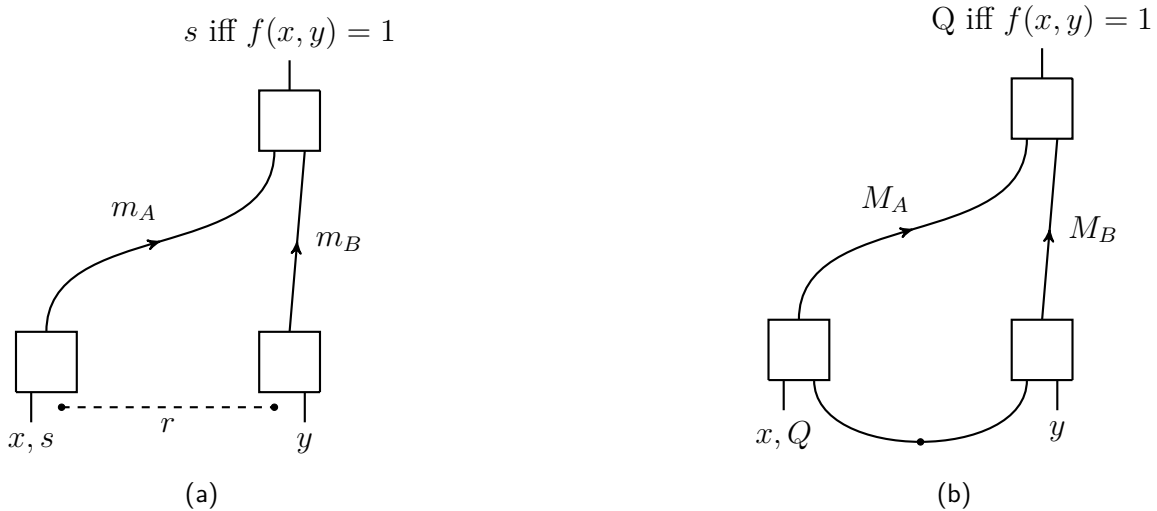

Here we further explore quantum CDS, with an emphasis on understanding the relationship between quantum and classical CDS.
We ask if quantum resources provide advantages in implementing CDS, and to what extent the same or analogous lower bounds apply to quantum CDS as to classical CDS. 

\subsection{Prior work}

Quantum CDS was introduced in \cite{allerstorfer2024relating}, where it was shown to be equivalent to $f$-routing, a class of non-local computations studied in the context of position-based cryptography \cite{kent2011quantum,buhrman2014position}. 
In an $f$-routing protocol Alice and Bob are given input strings $x\in X=\{0,1\}^n$ and $y\in Y= \{0,1\}^n$ respectively, and a quantum system $Q$ is held by Alice. 
Alice and Bob share a quantum state, and will act on their locally held systems to each produce two output systems. 
They each keep one of the output systems and send the other system to the other player. 
The desired functionality of the protocol is specified by a function $f:X\times Y\rightarrow \{0,1\}$: at the end of the protocol, Alice should hold $Q$ if $f(x,y)=0$ , while Bob should hold $Q$ if $f(x,y)=1$. 
Non-local quantum computation (NLQC, as first formalized by \cite{buhrman2014position}) can be seen as a generalization of this scenario where the inputs are fully quantum, and Alice and Bob should jointly enact a general quantum channel.  

The equivalence in \cite{allerstorfer2024relating} shows that an $f$-routing protocol for a function $f$ gives a quantum CDS protocol using essentially the same resources. 
Conversely, a quantum CDS protocol gives an $f$-routing protocol using similar resources, but where any random bits used in the CDS protocol become shared EPR pairs in the $f$-routing protocol. 
Further, classical CDS protocols imply quantum CDS protocols using similar resources. 
Taken together then, these results imply that upper bounds on classical CDS give upper bounds on quantum CDS and $f$-routing, and that lower bounds on quantum CDS or $f$-routing give lower bounds on classical CDS. 
One result that follows via these connections is a good upper bound on the communication and entanglement needed to perform quantum CDS for functions which can be computed by a quantum circuit in constant $T$-depth, which was known previously for $f$-routing \cite{speelman2015instantaneous}. 
We return to and use this result later. 

Beyond the immediate translation of results via the above implications, the connection between $f$-routing and classical CDS more broadly gives us a classical analogue and starting point for addressing open questions in non-local quantum computation. 
Indeed, NLQC is a poorly understood subject, with basic questions remaining unanswered despite considerable effort. 
Classical CDS provides a sometimes simpler starting point, and we can look for analogous properties or proofs in the case of NLQC. 
Aside from the interest in NLQC broadly, we are also interested in quantum CDS in particular. 
Indeed, we ask about the power of quantum resources in information-theoretic cryptography. 

\begin{table}
\centering
\begin{tabular}{|c|c|c|}
\hline
    & Classical & Quantum \\
   \hline
   perfectly secure & \makecell{ $\PP^{\cc}$} \cite{applebaum2021placing} & \makecell{$\QNP^{\cc}$ \cite{asadi2024rank},\\ $\PP^{\cc}$ \cite{asadi2024conditional}} \\
   \hline
   perfectly correct & $\coNP^{\cc}$ \cite{applebaum2021placing} & $\coQNP^{\cc}$ \cite{asadi2024rank} \\
   \hline
   robust & \makecell{$\log (\R_{A\rightarrow B})+\log(\R_{B\rightarrow A})$ \cite{gay2015communication} \\ $\IP[2]^{\cc}$ \cite{applebaum2021placing} \\$\HVSZK^{\cc}$\cite{applebaum2021placing} \\ $\AM^{\cc}$ \cite{applebaum2021placing}} & \makecell{ \textcolor{blue}{$\log(\R_{0,A\rightarrow B})+\log(\R_{0,B\rightarrow A})$}, $\log \Q^*_{A\rightarrow B}$ \cite{asadi2024conditional} \\ $\QIP[2]^{\cc}$ \cite{asadi2024conditional}\\ $\HVQSZK^{\cc}$ \cite{asadi2024conditional}\\ \textcolor{blue}{$\QAM[2,2]^{\cc}$} } \\
   \hline
\end{tabular}
\caption{Summary of lower bounds on quantum and classical CDS; those in black are from prior work. The bounds are closely analogous: three classical lower bounds are reproduced in the quantum setting, with the classical lower bounding class replaced by a quantum analogue of that class. We add to this analogy by adding two new lower bounds, shown in blue.}\label{fig:lowerboundstable}
\end{table}

As a first step in exploring these directions, \cite{asadi2024conditional} took up the systematic study of quantum CDS and established a number of basic results. 
In particular, they proved that the soundness and correctness parameters of quantum CDS can be efficiently amplified, that the cost of a CDS protocol doesn't grow too quickly when combining functions using small formulas, and began exploring lower bounds on quantum CDS. 
To do this, they worked from analogy with the classical setting, where a number of lower bounds have been proven based on measures of classical communication complexity.
These are summarized in table \ref{fig:lowerboundstable}, along with the quantum counterparts known so far. 
In particular, \cite{asadi2024conditional} proved:
\begin{itemize}
    \item A lower bound on robust CDS from quantum one-way communication complexity, mirroring the lower bound on classical CDS from classical one-way communication complexity.
    \item A lower bound on perfectly correct CDS from $\PP^{\cc}$ complexity, matching the classical lower bound for the same setting (that these bounds match is related to the fact that $\QPP^{\cc}$ $=$ $\PP^{\cc}$).
    \item A lower bound on robust CDS from two-message interactive quantum proofs, mirroring a lower bound from two-message interactive classical proofs in the classical setting. 
\end{itemize}
A missing part of the analogy was a lower bound on the classical setting from $\coNP^{\cc}$ complexity, for which \cite{asadi2024conditional} found no quantum analogue.
This was later established in \cite{asadi2024rank}, via a technique apparently unrelated to the classical one.
Intriguingly, this last lower bound also led to a new quantum lower bound on perfectly secure quantum CDS from $\QNP^{\cc}$. 
This also means classical perfectly secure CDS is lower bounded by $\QNP^{\cc}$ complexity, which represents a new insight into classical CDS. 

\subsection{Our results}

In this work, we further explore the analogies between quantum and classical CDS, with the goal of a deeper understanding of both the classical and quantum settings. 

Regarding the study of lower bounds, we first revisit the lower bound on quantum CDS from the one-way quantum communication complexity, and prove that this can be upgraded to a similar bound in terms of the classical one way communication complexity, 
\begin{align}
    \overline{\text{CDQS}}(f)=\tilde{\Omega}(\log (\R_{0,A\rightarrow B}(f))+\log (\R_{0,B\rightarrow A}(f))).
\end{align}
Here $\R_{0,A\rightarrow B}$ is the classical one-way (deterministic) communication complexity, $\overline{\text{CDQS}}(f)$ indicates the communication plus entanglement cost of quantum CDS for the function $f$, and $\tilde{\Omega}$ indicates we have ignored a dependence on $\log \overline{\text{CDQS}}(f)$. 
This essentially matches the lower bound on classical CDS \cite{gay2015communication}, 
\begin{align}
    \text{CDS}(f)=\Omega(\log (\R_{0,A\rightarrow B}(f))+\log(\R_{0,B\rightarrow A}(f))),
\end{align}
with the distinction that the bound is on the communication alone in the classical case. 
That these bounds (nearly) match is an indication that the lower bounds from one-way communication complexity on classical CDS are in some sense weak: the lower bound for classical CDS must not be fully exploiting the structure of a CDS protocol, since it already applies to a much broader class of protocols (those using quantum strategies). 

In the classical setting, the CDS complexity is lower bounded by \cite{applebaum2021placing}
\begin{align}
    \text{CDS}(f) =\Omega(\IP[2]^{\cc}(f)).
\end{align}
From this starting point, \cite{applebaum2021placing} adapts standard results on classical interactive proofs to the communication setting to transform this into a two-message public coin protocol and hence transform the above bound into
\begin{align}\label{eq:AMbound}
    \text{CDS}(f) \geq [\AM^{\cc}(f)]^\alpha - \text{polylog}(n)
\end{align}
where $\alpha$ is a constant. 
Following the analogous strategy in the quantum case, \cite{asadi2024conditional} obtained
\begin{align}
    \text{CDQS}(f) = \Omega(\QIP[2]^{\cc}(f)).
\end{align}
Continuing in analogy to the classical strategy however, when we apply standard transformations for quantum interactive proofs we obtain
\begin{align}
    \overline{\text{CDQS}}(f) = \Omega(\QMAM^{\cc}(f))
\end{align}
where a $\QMAM$ proof involves a message sent to the verifier from the prover, a single public coin sent to the prover, then a message sent to the verifier. 
Unfortunately, since $\QIP=\QIP[3]=\QMAM$ this is a \emph{weakening} of the bound from $\QIP[2]$. 

To obtain a non-trivial public coin bound, we look for a new reduction to a two round public coin protocol that doesn't take the reduction to $\text{QIP}[2]^{\cc}$ as its starting point. 
We find that such a reduction is possible, but at the expense of adding an additional prover, 
\begin{align}
    \overline{\text{CDQS}}(f) = \Omega(\QAM[2,2]^{\cc}(f))
\end{align}
where the right hand side denotes the communication cost of a two-message, two-prover, public coin proof. 
This bound appears to be the closest available quantum analogue to the classical lower bound from $\AM^{\cc}$. 
In the classical setting \cite{applebaum2021placing} pointed out that the lower bound $\text{CDS}(f) \geq \Omega([\AM^{\cc}(f)]^\alpha)$ reveals lower bounding classical CDS as a problem which is easier than but closely related to the long standing goal in communication complexity of finding lower bounds on $\AM^{\cc}$ complexity for explicit functions. 
Similarly, our lower bound suggests viewing lower bounds on quantum CDS as step towards lower bounds on $\QAM[2,2]^{\cc}$.

Aside from the study of lower bounds on quantum CDS, we also look for interesting upper bounds to establish quantum advantages.
Indeed we prove that quantum resources provide an advantage in some CDS settings. 
Concretely, considering perfectly correct CDS we prove a lower bound of $\Omega(n)$ for a promise version of the not-equals function in the classical setting, and an $O(\log n)$ upper bound using entanglement. 
The protocol is a variation of the standard strategy used to solve the Deutsch-Jozsa problem in the quantum communication complexity setting. 

In the robust setting (imperfect correctness and imperfect privacy), we have bounds that either 1) can be evaluated explicitly (those in terms of one-way communication complexity) but match between the classical and quantum cases or 2) bounds which may not match (those in terms of $\AM^{cc}$ or $\QAM[2,2]^{\cc}$ complexity) but cannot be evaluated explicitly, at least without breakthroughs in communication complexity.\footnote{In fact the classical lower bound can be framed in terms of $\AM^{\cc}\cap \mathsf{co}\AM^{\cc}$, which is the smallest class against which explicit bounds are not known in communication complexity.}
Consequently, we can't hope for unconditional quantum-classical separations in the robust setting given the current state of knowledge of classical lower bounds in communication complexity.
However, we explore the power of quantum resources in the robust setting by giving a $O(\log n)$ upper bound for the ``forrelation'' function \cite{AA15}, a partial function which has been important for establishing classical-quantum separations in other contexts. 
Since there is no known sub-linear classical upper bound for this function, this provides evidence for the power of quantum resources in robust CDS. 
The strategy for the protocol combines techniques from the non-local quantum computation literature with techniques from communication complexity. 
In particular, \cite{speelman2015instantaneous} showed how to do non-local computations with low T-depth efficiently, and \cite{girish2022quantum} proved classical-quantum communication separations in contexts where the quantum protocol has low complexity. 
Our protocol involves viewing the quantum CDS protocol as an instance of a non-local computation and implementing the low complexity protocol of \cite{girish2022quantum} using the low T-depth technique in \cite{speelman2015instantaneous}. 

\section{Background and tools}

\subsection{Some quantum information tools}

Let $\mathcal{D}(\mathcal{H}_A)$ be the set of density matrices on the Hilbert space $\mathcal{H}_A$. 
Given two density matrices $\rho$, $\sigma\in \mathcal{D}(\mathcal{H}_A)$,  define the fidelity,
\begin{align}
    F(\rho,\sigma) \equiv \left( \tr\left(\sqrt{\sqrt{\rho}\,\sigma\sqrt{\rho}}\right)\right)^2 \: ,
\end{align}
which is related to the one norm distance $\left\Vert\rho-\sigma\right\Vert_1$ by the Fuchs-van de Graaf inequalities, 
\begin{align}
    1- \sqrt{F(\rho,\sigma)} \leq \frac{1}{2}\left\Vert\rho-\sigma\right\Vert_1 \leq \sqrt{1-F(\rho,\sigma)} \: .
\end{align}
Next define the diamond norm distance, which is a distance measure on quantum channels. 
\begin{definition} Let $\mathbfcal{N}_{B\rightarrow C}, \mathbfcal{M}_{B\rightarrow C}: \mathcal{L}(\mathcal{H}_A)\rightarrow \mathcal{L}(\mathcal{H}_B)$ be quantum channels. 
The \textbf{diamond norm distance} is defined by 
\begin{align}
    \left\Vert \mathbfcal{N}_{B\rightarrow C}-\mathbfcal{M}_{B\rightarrow C}\right\Vert_\diamond = \sup_{d} \max_{\Psi_{A_dB}}\left\Vert \mathbfcal{N}_{B\rightarrow C}(\Psi_{A_dB}) - \mathbfcal{M}_{B\rightarrow C}(\Psi_{A_dB})\right\Vert_1
\end{align}
where $\Psi_{A_dB}\in \mathcal{D}(\mathcal{H}_{A_d}\otimes \mathcal{H}_B)$ and $\mathcal{H}_{A_d}$ is a $d$-dimensional Hilbert space. 
\end{definition}
The diamond norm distance has an operational interpretation in terms of the maximal probability of distinguishing quantum channels \cite{kitaev2002classical,wilde2013quantum}. 

From \cite{kretschmann2008continuity} we have the following theorem. 
\begin{theorem}\label{thm:operatorvsdiamond}
    For any two channels $\mathbfcal{T}_1$ and $\mathbfcal{T}_2$, 
\begin{align}\label{eq:opunderdiamond}
    \frac{\left\Vert\mathbfcal{T}_1-\mathbfcal{T}_2\right\Vert_\diamond}{\sqrt{\left\Vert\mathbfcal{T}_1\right\Vert_\diamond}+\sqrt{\left\Vert\mathbfcal{T}_2\right\Vert_\diamond}} \leq \inf_{\mathbf{V}_1,\mathbf{V}_2}\left\Vert\mathbf{V}_1-\mathbf{V}_2 \right\Vert_{\mathrm{op}} \leq \sqrt{\left\Vert\mathbfcal{T}_1-\mathbfcal{T}_2\right\Vert_{\diamond}} \: .
\end{align}
where the infimum is over isometric extensions of $\mathbfcal{T}_1$ and $\mathbfcal{T}_2$.
\end{theorem}
We will make use of the following remark. 
\begin{remark}
    For any two channels $\mathbfcal{T}_1$ and $\mathbfcal{T}_2$, we have that
    \begin{align}\label{eq:isometricextensionbound}
        \inf_{\mathbf{V}_1,\mathbf{V}_2}\left\Vert \mathbf{V}_1-\mathbf{V}_2\right\Vert_\diamond \leq 2\sqrt{\left\Vert \mathbfcal{T}_1-\mathbfcal{T}_2\right\Vert_\diamond}
    \end{align}
    where the infimum is over isometric extensions of the channels $\mathbfcal{T}_1$ and $\mathbfcal{T}_2$ labelled $\mathbf{V}_1$ and $\mathbf{V}_2$ respectively. 
\end{remark}
\begin{proof}
Consider the diamond norm $\left\Vert\mathbf{V}_1-\mathbf{V}_2\right\Vert_\diamond$ for any isometries $\mathbf{V}_1, \mathbf{V}_2$. 
Using the lower bound in equation \eqref{eq:opunderdiamond} and using that $\left\Vert\mathbf{V}_1\right\Vert_\diamond=\left\Vert\mathbf{V}_2\right\Vert_\diamond=1$, we have
\begin{align}\label{eq:stateprep1}
    \frac{1}{2}\left\Vert\mathbf{V}_1-\mathbf{V}_2\right\Vert_\diamond \leq \inf_{P_1, P_2}\left\Vert\mathbf{V}_1\otimes P_1-\mathbf{V}_2\otimes P_2 \right\Vert_{\mathrm{op}} 
\end{align}
where $P_1$ and $P_2$ are state preparation channels, and we have used that the only isometric extensions of isometries is to append a state preparation channel. 
Then, since taking the state preparation channels to be trivial is one possible choice of state preparation channel, we have
\begin{align}\label{eq:stateprep}
    \inf_{P_1, P_2}\left\Vert\mathbf{V}_1\otimes P_1-\mathbf{V}_2\otimes P_2 \right\Vert_{\mathrm{op}} \leq \left\Vert\mathbf{V}_1-\mathbf{V}_2\right\Vert_{\mathrm{op}} 
\end{align}
and hence, combining equations \eqref{eq:stateprep1} and \eqref{eq:stateprep}
\begin{align}
    \frac{1}{2}\left\Vert\mathbf{V}_1-\mathbf{V}_2\right\Vert_\diamond \leq \left\Vert\mathbf{V}_1-\mathbf{V}_2\right\Vert_{\mathrm{op}}\,.
\end{align}
Then since this was true for all isometries, we can combine it with the upper bound in \eqref{eq:opunderdiamond} to obtain
\begin{align}
    \frac{1}{2}\inf_{\mathbf{V}_1,\mathbf{V}_2}\left\Vert\mathbf{V}_1-\mathbf{V}_2\right\Vert_\diamond \leq \sqrt{\left\Vert\mathbfcal{T}_1-\mathbfcal{T}_2 \right\Vert_\diamond}
\end{align}
as needed. 
\end{proof}

Given a quantum channel $\mathbfcal{N}_{A\rightarrow B}$, we define a complementary channel $(\mathbfcal{N})^c_{A\rightarrow C}$ as any channel such that there exists an isometry $\mathbf{V}_{A\rightarrow BC}$ such that
\begin{align}
    \mathbfcal{N}_{A\rightarrow B} (\cdot)&= \tr_C(\mathbf{V}_{A\rightarrow BC} \,(\cdot) \mathbf{V}^\dagger_{A\rightarrow BC}) \nonumber \\
    (\mathbfcal{N})^c_{A\rightarrow C}(\cdot) &= \tr_B(\mathbf{V}_{A\rightarrow BC}  \,(\cdot)\mathbf{V}^\dagger_{A\rightarrow BC}) \: .
\end{align}

We will use the following bound from \cite{afham2022quantum}. 
Suppose we have an ensemble of mixed states, $\{(p_i, \rho_i)\}$.
Then
\begin{align}\label{eq:handyFbound}
    \max_{\sigma} \sum_i p_i \sqrt{F(\sigma, \rho_i)} \leq \sqrt{\sum_{i,j}p_ip_j \sqrt{F(\rho_i,\rho_j)}}.
\end{align}
In words, if the $\rho_i$ in the ensemble are very different then we can't choose a $\sigma$ that is close to all of them at once. 

\subsection{Definition of CDS and some basic properties}

We begin by defining the classical CDS setting. 
\begin{definition}\label{def:CDS}
    A \textbf{conditional disclosure of secrets (CDS)} task with classical resources is defined by a choice of function $f:\{0,1\}^{2n}\rightarrow \{0,1\}$.
    The scheme involves input $x\in \{0,1\}^{n}$ given to Alice and input $y\in \{0,1\}^{n}$ given to Bob.
    Alice and Bob share a random string $r\in R$.
    Additionally, Alice holds a string $s$ drawn from a distribution $S$, which we call the secret. 
    Alice sends message $m_A(x,s,r)\in M_A$ to the referee, and Bob sends message $m_B(y,r)\in M_B$.  
    We require the following two conditions on a CDS protocol. 
    \begin{itemize}
        \item $\epsilon$\textbf{-correct:} There exists a decoding function $D(m_A,x,m_B,y)$ such that 
        \begin{align}
            \forall s\in S,\,\forall \,(x,y) \in X\times Y \,\,s.t.\,\,f(x,y)=1,\,\,\, \underset{r\leftarrow R}{\mathrm{Pr}}[D(m_A,x,m_B,y)=s] \geq 1-\epsilon \: .
        \end{align}
        \item $\delta$\textbf{-secure:} There exists a simulator producing a distribution $Sim$ taking on values in $M=M_AM_B$ such that
        \begin{align}
            \forall s\in S,\,\forall \,(x,y) \in X\times Y \,\,s.t.\,\, f(x,y)=0,\,\,\, \left\Vert  Sim_{M|xy} - P_{M|xys} \right\Vert _1\leq \delta \: .
        \end{align}
    \end{itemize}
\end{definition}

We define the communication cost of a CDS protocol to be
\begin{align}
    t = \log |M_A| + \log |M_B| \: ,
\end{align}
where the logarithms are always taken base 2. 
For messages encoded into bits, this is the total number of bits of communication from Alice and Bob to the referee in a given protocol. 
Specifically, we maximize $t$ over choices of input $x,y$. 
The minimal communication cost for a function $f$ that achieves $\epsilon$-correctness and $\delta$-security we denote by $\text{CDS}_{\epsilon, \delta}(f)$. 
We denote the minimal number of shared random bits needed by $\overline{\text{CDS}}_{\epsilon, \delta}(f)$.
We will also use shorthand $\text{CDS}(f) = \text{CDS}_{0.09, 0.09}(f)$, $\overline{\text{CDS}}(f) = \overline{\text{CDS}}_{0.09, 0.09}(f)$.\footnote{As we comment below, the choice of default errors $\epsilon = \delta = 0.09$ is motivated by the values for which an amplification result can be shown in the quantum setting.} 

By default, we will assume single-bit secrets ($\mathrm{supp}(S)=\{0,1\}$) when discussing CDS complexity. 
Note that this is essentially without loss of generality: \cite{applebaum2021conditional} characterize the relationship between $\CDS(f)$ with single-bit and multi-bit secrets for a given $f$. 
Namely, that a protocol for single-bit secrets can be extended to $k$-bit secrets with an $O(k)$ multiplicative overhead to communication and randomness complexity; this transformation furthermore amplifies both privacy and correctness exponentially in $k$ (\cite{applebaum2021conditional}, theorem 2.3).

We will be especially interested in comparing properties of classical CDS with properties of quantum CDS, which we define next. 
\begin{definition}\label{def:CDQS}
    A \textbf{conditional disclosure of secrets task with quantum resources (CDQS)} is defined by a choice of function $f:\{0,1\}^{2n}\rightarrow \{0,1\}$, and a $d_Q$-dimensional Hilbert space $\mathcal{H}_Q$ which holds the secret.
     The task involves inputs $x\in \{0,1\}^{n}$ and system $Q$ given to Alice, and input $y\in \{0,1\}^{n}$ given to Bob.
    Alice sends message system $M_A$ to the referee, and Bob sends message system $M_B$. 
    Alice and Bob share a resource state $\Psi_{LR}$ with $L$ held by Alice and $R$ held by Bob. 
    Label the combined message systems as $M=M_AM_B$.
    Label the quantum channel defined by Alice and Bob's combined actions $\mathbfcal{N}_{Q\rightarrow M}^{x,y}$.
    We put the following two conditions on a CDQS protocol. 
    \begin{itemize}
        \item $\epsilon$\textbf{-correct:} There exists a channel $\mathbfcal{D}^{x,y}_{M\rightarrow Q}$, called the decoder, such that
        \begin{align}
            \forall (x,y)\in X\times Y \,\,\, s.t. \,\, f(x,y)=1,\,\,\, \left\Vert\mathbfcal{D}^{x,y}_{M\rightarrow Q}\circ \mathbfcal{N}^{x,y}_{Q\rightarrow M} - \mathbfcal{I}_{Q\rightarrow Q}\right\Vert_\diamond \leq \epsilon \: .
        \end{align}
        \item $\delta$\textbf{-secure:} There exists a quantum channel $\mathbfcal{S}_{\varnothing \rightarrow M}^{x,y}$, called the simulator, such that
        \begin{align}
            \forall (x,y)\in X\times Y \,\,\, s.t. \,\, f(x,y)=0,\,\,\, \left\Vert \mathbfcal{S}_{\varnothing \rightarrow M}^{x,y} \circ \tr_Q - \mathbfcal{N}_{Q\rightarrow M}^{x,y}\right\Vert_\diamond \leq \delta \: .
        \end{align}
    \end{itemize}
\end{definition}
We will take the Hilbert space $Q$ to be 2 dimensional throughout this work. 
The communication pattern of a CDQS protocol is shown in figure \ref{fig:CDQSprotocol}. 
We define the communication cost of a CDQS protocol to be
\begin{align}
    t = \log \text{dim}(M_A) + \log \text{dim}(M_B) \: .
\end{align}
For qubit systems this is the total number of qubits of communication from Alice and Bob to the referee.
We maximize the above over choices of input $x,y$. 
The minimal communication cost for a function $f$ that achieves $\epsilon$-correctness and $\delta$-security we denote by $\text{CDQS}_{\epsilon, \delta}(f)$. 
We denote the minimal number of qubits in the shared resource system plus the qubits of message by $\overline{\text{CDQS}}_{\epsilon, \delta}(f)$.\footnote{Notice that this notation differs from the classical case, where the overline indicates just the randomness. In the classical case the randomness lower bounds the communication \cite{applebaum2021placing}, while in the quantum case we don't know a similar statement. This discrepancy leads to the different notations being natural in the two contexts.}
We will also use $\text{CDQS}(f) = \text{CDQS}_{0.09, 0.09}(f)$, $\overline{\text{CDQS}}(f) = \overline{\text{CDQS}}_{0.09, 0.09}(f)$.

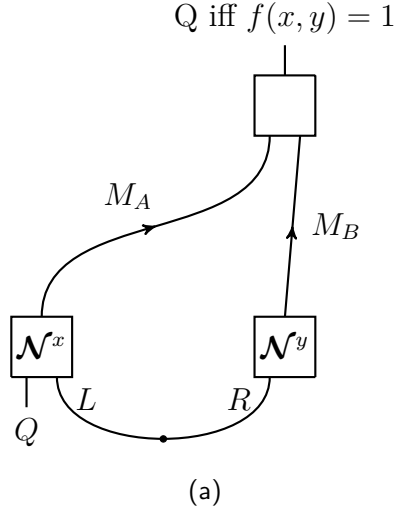
\begin{figure*}
    \centering
    \begin{subfigure}{0.45\textwidth}
    \centering
        \begin{tikzpicture}[scale=0.4]
    
    \draw[thick] (-5,-5) -- (-5,-3) -- (-3,-3) -- (-3,-5) -- (-5,-5);
    \node at (-4,-4) {$\mathbfcal{N}^x$};
    
    \draw[thick] (5,-5) -- (5,-3) -- (3,-3) -- (3,-5) -- (5,-5);
    \node at (4,-4) {$\mathbfcal{N}^y$};
    
    \draw[thick] (5,5) -- (5,3) -- (3,3) -- (3,5) -- (5,5);
    
    \draw[thick, mid arrow] (4,-3) -- (4.5,3);
    
    \draw[thick, mid arrow] (-4,-3) to [out=90,in=-90] (3.5,3);
    
    \draw[thick] (-3.5,-5) to [out=-90,in=-90] (3.5,-5);
    \draw[black] plot [mark=*, mark size=3] coordinates{(0,-7.05)};
    \node[below right] at (-3.25,-5) {$L$};
    \node[below left] at (3.25,-5) {$R$};

    \node[left] at (0,1) {$M_A$};
    \node[right] at (4.5,0) {$M_B$};
    
    \draw[thick] (-4.5,-6) -- (-4.5,-5);
    \node[below] at (-4.5,-6) {$Q$};
    
    \draw[thick] (4,5) -- (4,6);
    \node[above] at (4,6) {Q iff $f(x,y)=1$};
    
    \end{tikzpicture}
    \caption{}
    \end{subfigure}
    \caption{A CDQS protocol, with system labels and location of each quantum operation. The density matrix on $M_AM_B$ we refer to as the mid-protocol density matrix. We sometimes combine the actions of Alice and Bob to define $\mathbfcal{N}^{x,y}_{Q\rightarrow M}=\mathbfcal{N}^x_{AL\rightarrow M_A}\otimes\mathbfcal{N}^y_{R\rightarrow M_B}$.} 
    \label{fig:CDQSprotocol}
\end{figure*}

Note that for quantum CDS, whenever $\epsilon, \delta\leq 0.09$ we can amplify and achieve parameters $\epsilon'=\epsilon 2^{-k}$, $\delta'=\delta 2^{-k}$ using an overhead in communication and entanglement of a factor of $k$.
More precisely, we have the following theorem from \cite{asadi2024conditional}:
\begin{theorem}\label{thm:amplify}
    Let $F_Q$ be a $\CDQS$ protocol for a function $f$ that supports one qubit  secrets with correctness error $\delta=0.09$ and privacy error $\epsilon=0.09$, has communication cost $c$, and entanglement cost $E$. 
    Then for every integer $k$, there exists a $\CDQS$ protocol $G_Q$ for $f$ with $k$-qubit secrets, privacy and correctness errors of $2^{-\Omega(k)}$, and communication and entanglement complexity of size $O(k c)$ and $O(k E)$, respectively. 
\end{theorem}
Note that in this result we increase the size of the secret while simultaneously amplifying correctness and security. 

We will also make use of the following result about CDQS. 
\begin{lemma}\label{lemma:complementaryCDS}
    Suppose that a CDQS protocol is $\delta$ secure, and denote the encoding map by $\mathbfcal{N}_{Q\rightarrow M}^{x,y}$. 
    Then for $(x,y)\in f^{-1}(0)$ there exists a decoding map $\mathbfcal{D}^{x,y}$ such that
    \begin{align}
        \left\Vert\mathbfcal{D}^{x,y}_{M'\rightarrow Q}\circ (\mathbfcal{N}^{x,y})^c_{Q\rightarrow M'} - \mathbfcal{I}_Q \right\Vert_\diamond \leq 2\sqrt{\delta} \: .
    \end{align}
\end{lemma}
This lemma follows from the proof of theorem 23 in \cite{allerstorfer2024relating}. 
Briefly, the intuition behind this result comes from the decoupling theorem: if quantum information is not recoverable from the channel $\mathbfcal{N}^{x,y}_{Q\rightarrow M}$, then, since information is not destroyed in quantum mechanics, it must be recoverable from the environment system and hence from the output of any complementary channel. 

CDS is related to another primitive known as private simultaneous message passing, which we define as follows. 
\begin{definition}\label{def:PSM}
    A \textbf{private simultaneous message (PSM)} task is defined by a choice of function $f:X\times Y\rightarrow Z$. 
    The inputs to the task are $n$ bit strings $x\in X$ and $y\in Y$ given to Alice and Bob, respectively. 
    Alice then sends a message $m_A(x,r)$ to the referee, and Bob sends message $m_B(y,r)$. 
    From these inputs, the referee prepares an output bit $z\in Z$. 
    We require the task be completed in a way that satisfies the following two properties. 
    \begin{itemize}
        \item \textbf{$\epsilon$-correctness:} There exists a decoder $Dec$ such that \begin{align}
            \forall (x,y)\in X\times Y,\, \,\,\,\mathrm{Pr}[Dec(m_A,m_B)=f(x,y)] \geq 1-\epsilon \: .
        \end{align}
        \item \textbf{$\delta$-security:} There exists a simulator producing a distribution $Sim$ taking on values in $M=M_AM_B$, such that 
        \begin{align}
            \forall (x,y)\in X\times Y,\, \,\,\,\left\Vert Sim_{M|f(x,y)} - P_{M|xy}\right\Vert_1\leq \delta \: .
        \end{align}
        Stated differently, the distribution of the message systems is $\delta$-close to one that depends only on the function value, for every choice of $x,y$. 
    \end{itemize}
\end{definition}
PSM is a stronger primitive than CDS in that a PSM protocol for a function $f$ implies the existence of a CDS protocol for the same function with similar efficiency. 
This transformation is given, for example, as theorem 2.5 of \cite{applebaum2017private}

Next, we give the quantum definition. 
We follow the definition of \cite{allerstorfer2024relating}. 
\begin{definition}\label{def:PSQM}
    A \textbf{private simultaneous quantum message (PSQM)} task is defined by a choice of function $f:X\times Y\rightarrow Z$. 
    The inputs to the task are $n$ bit strings $x\in X$ and $y\in Y$ given to Alice and Bob, respectively. 
    Alice then sends a quantum message system $M_A$ to the referee, and Bob sends quantum message system $M_B$. 
    From the combined message system $M=M_AM_B$, the referee prepares an output qubit on system $Z$. 
    We require the task be completed in a way that satisfies the following two properties.
    \begin{itemize}
        \item \textbf{$\epsilon$-correctness:} There exists a decoding map $\mathbf{V}_{M \rightarrow Z\tilde{M}}$ such that 
        \begin{align}
            \forall (x,y)\in X\times Y, \,\,\,\,\, \left\Vert\tr_{\tilde{M}}(\mathbf{V}_{M \rightarrow Z\tilde{M}} \rho_{M}(x,y) \mathbf{V}_{M \rightarrow Z\tilde{M}}^\dagger ) - \ketbra{f(x,y)}{f(x,y)}_Z\right\Vert_1 \leq \epsilon \: .
        \end{align}
        where $\rho_M(x,y)$ is the density matrix on $M$ produced on inputs $x,y$.
        \item \textbf{$\delta$-security:} There exists a simulator, which is a quantum channel $\mathbfcal{S}_{Z\rightarrow M}(\cdot)$, such that 
        \begin{align}
            \forall (x,y)\in X\times Y,\,\,\,\,\,\left\Vert\rho_{M}(x,y) - \mathbfcal{S}_{Z\rightarrow M}(\ketbra{f(x,y)}{f(x,y)}_Z)\right\Vert_1 \leq \delta \: .
        \end{align}
        Stated differently, the state of the message systems is $\delta$-close to one that depends only on the function value, for every choice of input.
    \end{itemize}
\end{definition}
As in the classical case, a quantum PSM protocol for the function $f$ gives a good quantum CDS protocol for the same function with similar efficiency \cite{allerstorfer2024relating}. 

One important difference between CDS and PSM is that PSM is lower bounded linearly by the simultaneous message passing model in communication complexity, whereas the best bounds on CDS in terms of communication complexity (in the case where finite errors are allowed) are logarithmic. 
The key distinction is that in CDS the referee knows the inputs $x,y$, whereas in PSM the referee does not know the inputs. 

\section{Revisiting lower bounds from communication complexity}

\subsection{Lower bounds from one-way communication complexity}

In this section we will give a lower bound on quantum CDS in terms of the one-way classical communication complexity. 
We first recall how this is defined. 

\begin{definition}[Classical one-way communication complexity] Let $f:\{0,1\}^n\times \{0,1\}^n\rightarrow \{0,1\}$ and $\delta\in[0,1]$.
A one-way communication protocol for $f$ is defined as follows.
Alice receives $x \in\{0,1\}^n$ as input and produces a classical string $m_A$ as output, which she sends to Bob. 
Bob receives $y \in\{0,1\}^n$ and $m_A$, and outputs a bit $z$.
The protocol is $\delta$-correct if $\Pr[z=f(x,y)]\geq 1-\delta$. 

The classical one-way communication complexity of $f$, $\R_{\delta,A\rightarrow B}(f)$ is defined as the minimum number of bits in $m_A$ needed to achieve $\delta$-correctness. 
We write $\R_{A\rightarrow B}(f)\equiv \R_{\delta=0.09,A\rightarrow B}(f)$. $\R_{B\rightarrow A}(f)$ is defined analogously by interchanging the roles of Alice and Bob.
\end{definition}

We also recall the definition of a stronger, quantum communication complexity class.
\begin{definition}[Quantum one-way communication complexity]
    Let $f:\{0,1\}^n\times \{0,1\}^n\rightarrow \{0,1\}$ and $\delta\in[0,1]$. A one-way quantum communication protocol for $f$ is defined as follows.
    At the start, Alice and Bob are each given the respective halves of a (potentially entangled) resource state $\Psi_{LR}$, and classical inputs $x\in\{0,1\}^n$ and $y\in\{0,1\}^n$. Alice produces a quantum message $M_A$ and sends it to Bob, who then outputs a bit $z$. The protocol is $\delta$-correct if $\Pr[z=f(x,y)]\geq1-\delta$.

    The quantum one-way communication complexity of $f$ with no preshared entanglement, $\Q_{\delta,A\rightarrow B}(f)$, is defined as the minimum number of qubits in $M_A$ to achieve $\delta$-correctness, when $\Psi_{LR}$ is restricted to a product state. Alternatively, $\Q^*_{\delta,A\rightarrow B}(f)$ is the minimum number of qubits in $M_A$ and $\Psi_{LR}$ to achieve $\delta$-correctness for general $\Psi_{LR}$. We write $\Q_{A\rightarrow B}$ and $\Q^*_{A\rightarrow B}$ for $\Q_{0.09,A\rightarrow B}$ and $\Q^*_{0.09,A\rightarrow B}$, respectively. The classes $\Q_{B\rightarrow A}$ and $\Q^*_{B\rightarrow A}$ are defined analogously by interchanging the roles of Alice and Bob.
\end{definition}
To relate this to quantum CDS, we use the following lemma, reproduced from \cite{asadi2024conditional}.\footnote{Similar observations appear earlier in the $f$-routing literature, going back to \cite{buhrman2013garden}.} 
The lemma captures a basic consequence of correctness and security of the CDQS protocol for the structure of the `mid-protocol density matrix', which is the state on Alice and Bob's messages systems along with a reference system.
\begin{lemma}\label{lemma:productness} \textbf{(Reproduced from \cite{asadi2024conditional})}
    Consider the mid-protocol density matrix of an $\epsilon$-correct, $\delta$-secure CDQS protocol whose $d_Q$-dimensional secret is taken to be a maximally entangled state between $Q$ and reference system $\bar{Q}$, i.e.
    \begin{align}
        \rho_{\bar{Q}M}(x,y)=\mathbfcal{N}_{Q\rightarrow M}^{x,y}(\Psi^+_{Q\bar{Q}}) \: .
    \end{align}
    where $\mathbfcal{N}_{Q\rightarrow M}^{x,y}$ represents the combined actions of Alice and Bob's operations. 
    Then, when $f(x,y)=0$ we have that for $\pi_{\bar{Q}}=\mathcal{I}/d_{Q}$
    \begin{align}
        \left\Vert\rho_{\bar{Q} M}(x,y)-\pi_{\bar{Q}} \otimes \rho_{M}(x,y) \right\Vert_1 \leq \delta \: ,
    \end{align}
    and when $f(x,y)=1$, we have that for all density matrices $\sigma_{\bar{Q}}, \sigma_{M}$,
    \begin{align}
        \left\Vert\rho_{\bar{Q}M}(x, y) - \sigma_{\bar{Q}}\otimes \sigma_{M} \right\Vert_1 \geq 2\left(1 - \frac{1}{\sqrt{d_Q}}\right) - \epsilon \: .
    \end{align}
\end{lemma}

In \cite{asadi2024conditional}, the authors prove the following lower bound on quantum CDS. 
\begin{theorem} \label{thm:oneway}[Reproduced from \cite{asadi2024conditional}]
The one-way quantum communication complexity of $f$ and the communication cost of a CDQS protocol for $f$ are related by
\begin{align}
    q_A+q_B = \Omega(\log \Q_{B\rightarrow A}^{*}(f))\,,
\end{align}
where $q_A$ is the number of qubits sent from Alice to the referee, and $q_B$ is the number of qubits sent from Bob to the referee. 
\end{theorem}

The proof idea is as follows. 
Starting with a quantum CDS protocol, we build a one-way quantum communication protocol. 
We have Alice and Bob share the same entangled state $\Psi_{LR}$ as in the CDQS protocol. 
Then, Alice and Bob perform the CDQS operations, call them $\mathbfcal{N}^x_{Q L\rightarrow M_A}$ and $\mathbfcal{N}^y_{R\rightarrow M_B}$, taking the secret $Q$ to be maximally entangled with a reference $\bar{Q}$, and Bob sends his message $M_{B}$ to Alice. 
Note that from lemma \ref{lemma:productness}, the resulting density matrix $\rho_{\bar{Q}M}(x,y)$ will be close to product across $\bar{Q}$ and $M=M_AM_B$ when $f(x,y)=0$ and close to maximally entangled when $f(x,y)=1$.
Repeating this procedure $2^{q_A+q_B}$ times, Alice can use standard tomography techniques to make measurements characterizing the density matrix and hence determine $f(x,y)$. 
Thus $2^{q_A+q_B}\geq \Q^*_{A\rightarrow B}(f)$, leading to the claimed bound. 

Our observation here is that we can adjust this strategy to get a lower bound from the classical communication complexity. 
The strategy is for Bob to send Alice a classical description of the quantum state $\rho_{LM_B}(y)=\mathbfcal{N}^y_{R\rightarrow M_B}(\Psi_{LR})$ rather than the state itself. 
This description is never too much larger than $2^{q_B + E}$ bits, leading to the following bound. 

\begin{theorem} \label{thm:class_ow}
Consider a robust CDQS protocol which uses a resource state $\Psi_{LR}$ with $L$ and $R$ consisting of $E$ qubits, and where Alice and Bob send $q_A$ and $q_B$ qubits to the referee respectively. 
Then,
\begin{align}
    q_{B} + E \geq \tilde{\Omega}(\log \R_{0, B\rightarrow A}(f))
\end{align}
The $\tilde{\Omega}$ notation indicates we have suppressed a dependence on $\log (q_B+E)$. 
The same bound also holds with $A\leftrightarrow B$. 
\end{theorem}

\begin{proof}
    Alice and Bob hold descriptions of the CDQS protocol (but need not hold the resource state $\Psi_{LR}$). 
    Upon receiving $y$, Bob sends to Alice a classical description of the state
    \begin{align}
        \rho_{LM_B}(y) = \mathbfcal{N}^y_{R\rightarrow M_B}(\Psi_{LR})
    \end{align}
    which he can compute, since he knows both $y$ and the description of the CDQS protocol. 
    We have him specify the entries in $\rho_{LM_B}$ to $k$ digits, where we choose $k$ later, so that his message is $k d_L^2d_B^2$ bits. 
    This describes a matrix $\hat{\rho}_{LM_B}$ which has each entry differ from the corresponding entry of $\rho_{LM_B}$ by at most $1/2^k$, so that
    \begin{align}
        \left\Vert\hat{\rho}_{LM_B}-\rho_{LM_B}\right\Vert_2 = \sqrt{\sum_{i,j=1}^{d_{L}d_{M_B}}|(\hat{\rho}_{LM_B}-\rho_{LM_B})_{ij}|^2} \leq \frac{d_Ld_{M_B}}{2^k}
    \end{align}
    where $(\hat{\rho}_{LM_B}-\rho_{LM_B})_{ij}$ denotes the $(i,j)$ matrix entry of $\hat{\rho}_{LM_B}-\rho_{LM_B}$. 
    We can relate this to the trace distance using that, for a $d\times d$ matrix, $\left\Vert A\right\Vert_1\leq \sqrt{d}\left\Vert A \right\Vert_2$, which here gives
    \begin{align}
        \left\Vert\hat{\rho}_{LM_B}- \rho_{LM_B}\right\Vert_1 \leq \sqrt{d_Ld_{M_B}}\left\Vert \hat{\rho}_{LM_B}- \rho_{LM_B}\right\Vert_2 \leq \frac{d_L^{3/2}d^{3/2}_{M_B}}{2^k}.
    \end{align}
    After Bob communicates his description of $\hat{\rho}$ to Alice, Alice will compute the density matrix $\rho_{\bar{Q}M}$ and check if it is close to product or not. 
    From lemma \ref{lemma:productness}, this allows her to determine $f(x,y)$. 
    Note that since Alice knows the channel $\mathbfcal{N}^x_{LQ\rightarrow M_A}$ this doesn't introduce any additional error, 
    \begin{align}
        ||\hat{\rho}_{QM}-\rho_{QM}||_1 \leq ||\hat{\rho}_{LM_B}-\rho_{LM_B}||_1.
    \end{align}
    Quantitatively, one can check that for Alice to be able to determine $||\rho_{\bar{Q}M}-\pi_{\bar{Q}}\otimes \rho_M||_1$ with sufficient precision it suffices for her to learn $\rho$ to within trace distance 
    \begin{align}
        ||\hat{\rho}_{\bar{Q}M}-\rho_{\bar{Q}M}||_1=\gamma(\epsilon, \delta) = \frac{1}{2}\left(1-\frac{1}{\sqrt{d_Q}}\right)- \frac{\epsilon}{4} - \frac{\delta}{4}.
    \end{align}
    Then, if the distance to the product state $\pi_{\bar{Q}}\otimes \rho_{M}$ is less than $\gamma$ she can conclude $f(x,y)=0$ with certainty, while if it is larger than that she can conclude $f(x,y)=1$ with certainty. 
    In terms of our parameter $k$, we need then that
    \begin{align}
        \frac{d_L^{3/2}d^{3/2}_{M_B}}{2^k} \leq \gamma(\epsilon, \delta)
    \end{align}
    so that $k=\frac{3}{2}(q_B+E)-\log\gamma(\epsilon, \delta)$ suffices. 
    Bob's total communication is $k\times d_L^2d_{M_b}^2$, so that
    \begin{align}
        \R_{0,B\rightarrow A}(f) \leq \left(\frac{3}{2}(q_B+E)-\log\gamma(\epsilon, \delta)\right)2^{2(E+q_B)}
    \end{align}
    or
    \begin{align}
        q_B+E \geq \frac{1}{2}\log \R_{0,B\rightarrow A} - O(\log(q_B+E))
    \end{align}
\end{proof}

We also observe the following simple corollary of this theorem. 

\begin{corollary} A CDQS protocol which uses $E$ qubits of shared resource, $q_A$ qubits of message from Alice, and $q_B$ qubits of message from Bob must satisfy
\begin{align}
    \overline{CDQS}(f)=2E+q_A+q_B \geq \tilde{\Omega}(\log \R_{0,B\rightarrow A}(f) + \log \R_{0,A\rightarrow B}(f))
\end{align}
Here, the $\tilde{\Omega}$ notation indicates that we've suppressed a dependence on $\log (q_A+E)$ and $\log( q_{B}+E)$.
\end{corollary}
Notice that this gives a different lower bound compared to $\log \Q^*_{B\rightarrow A}$, which in general is smaller than the lower bound above but also lower bounds the communication cost $q_A+q_B$ alone, rather than the entanglement cost plus the communication cost. 
Unlike in the classical case we do not in general know if the entanglement cost can be much larger than the communication cost, so a priori these are different bounds. 
However, in all known protocols the entanglement and communication costs are similar. 
In that setting the lower bound from classical communication complexity is stronger. 

\subsection{Two-prover, public coin lower bound}

In the classical setting, CDS can be lower bounded polynomially by the $\AM^{\cc}$ complexity, as we noted in equation \eqref{eq:AMbound}. 
To obtain a non-trivial public coin lower bound in the quantum case, we find it necessary to consider two-prover proof settings. 
We define the appropriate two-prover proof setting next. 

\begin{definition} \textbf{(Two-prover, two-message, public coin proof)} Let $f:\{0,1\}^n\times \{0,1\}^n\rightarrow \{0,1\}$ and $\epsilon, \delta\in (0,1)$. 
A two-prover, two-message, public coin proof for $f$ in the communication complexity setting is an interactive protocol executed by two provers, prover 1 and prover 2, and two verifiers, Alice and Bob.
Provers 1 and 2 both hold strings $x,y\in \{0,1\}^n$ and begin with a shared entangled state $\varphi_{PP'}^{x,y}$, with prover 1 holding $P$ and prover $2$ holding $P'$. 
Alice knows input $x\in \{0,1\}^n$, Bob knows input $y\in \{0,1\}^n$, and Alice and Bob share input state $\ket{\Psi}_{LR}$ which is independent of $x,y$.
The protocol proceeds as follows. 
\begin{itemize}
    \item Alice shares random bits $r\in \{0,1\}^{|r|}$ with $P$. 
    \item Prover 1 prepares systems $M=M_AM_B$ from system $P$ and sends message systems $M_A$ to Alice and $M_B$ to Bob.
    \item Prover 2 prepares systems $M'=M_A'M_B'$ from system $P'$ and sends message systems $M_A'$ to Alice and $M_B'$ to Bob. 
    \item Alice and Bob apply local operations, which may depend on $x$ and $y$ respectively, and communicate with one another. After this interaction round, Alice outputs either $0$ (reject) or $1$ (accept). 
\end{itemize}
We require that
\begin{itemize}
    \item \textbf{$\epsilon$-correctness}: For all $(x,y)\in f^{-1}(1)$, Alice accepts with probability at least $1-\epsilon$.
    \item \textbf{$\delta$-security}: For all $(x,y)\in f^{-1}(0)$, Alice accepts with probability at most $\delta$.
\end{itemize}
\end{definition}
The cost of a two-prover, two-message, public coin proof is defined as the total number of qubits of communication sent by the provers plus the total communication used by Alice and Bob. 
The minimal cost over all protocols for a function $f$ that achieves $\epsilon$-correctness and $\delta$-security we label as $\text{QAM}[2,2]_{\epsilon, \delta}^{\cc}(f)$. 

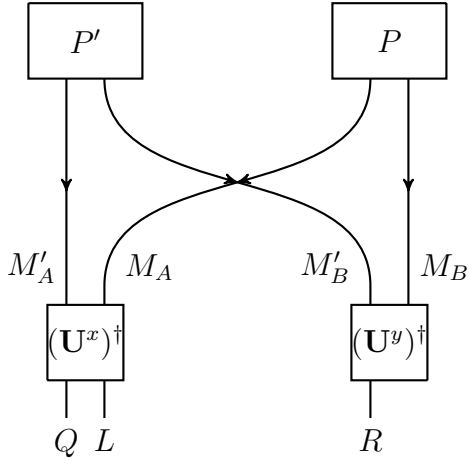
\begin{figure*}
    \centering
    \begin{tikzpicture}[scale=0.5]
    
    \draw[thick] (-5,-5) -- (-5,-3) -- (-3,-3) -- (-3,-5) -- (-5,-5);
    \node at (-4,-4) {$(\mathbf{U}^x)^\dagger$};
    
    \draw[thick] (5,-5) -- (5,-3) -- (3,-3) -- (3,-5) -- (5,-5);
    \node at (4,-4) {$(\mathbf{U}^y)^\dagger$};
    
    \draw[thick] (5.5,5) -- (5.5,3) -- (2.5,3) -- (2.5,5) -- (5.5,5);
    \node at (4,4) {$P$};
    
    \draw[thick] (-5.5,5) -- (-5.5,3) -- (-2.5,3) -- (-2.5,5) -- (-5.5,5);
    \node at (-4,4) {$P'$};
    
    \draw[thick, mid arrow] (-4.5,3) -- (-4.5,-3);
    \node[left] at (-4.5,-2) {$M_A'$};
    
    \draw[thick, mid arrow] (4.5,3) -- (4.5,-3);
    \node[right] at (4.5,-2) {$M_B$};
    
    \draw[thick, mid arrow] (3.5,3) to [out=-90,in=90] (-3.5,-2.5);
    \draw[thick] (-3.5,-2.5) -- (-3.5,-3);
    \node[right] at (-3.25,-2) {$M_A$};
    
    \draw[thick, mid arrow] (-3.5,3) to [out=-90,in=90] (3.5,-2.5);
    \draw[thick] (3.5,-2.5) -- (3.5,-3);
    \node[left] at (3.25,-2) {$M_B'$};
    
    \draw[thick] (-3.5,-6) -- (-3.5,-5);
    \node[below] at (-3.5,-6) {$L$};
    \draw[thick] (3.5,-6) -- (3.5,-5);
    \node[below] at (3.5,-6) {$R$};
    
    \draw[thick] (-4.5,-6) -- (-4.5,-5);
    \node[below] at (-4.5,-6) {$Q$};
    
    \end{tikzpicture}
    \caption{The two-prover proof protocol. Alice and Bob receive systems $M_A$ and $M_B$ from prover $P$, and systems $M_A'$ and $M_B'$ from prover $P'$. Alice applies $(\mathbf{U}^x_{M_AM_A'\rightarrow QL})^\dagger$, Bob applies $(\mathbf{U}^y_{R\rightarrow M_BM_B'})^\dagger$. Bob then sends $R$ to Alice, who measures $LR$ and $Q$ to check they are the inputs to the corresponding CDQS protocol.}
    \label{fig:CDSandtwoprovers}
\end{figure*}

We can now state our public-coin lower bound on robust quantum CDS.
\begin{theorem}\label{thm:QAMbound}
For any fixed $\epsilon_p, \delta_p$, 
\begin{align}
    \overline{\textnormal{CDQS}}(f) \geq \QAM[2,2]_{\epsilon_p,\delta_p}^{\cc}(f)  - c(\epsilon_p,\delta_p) \: .
\end{align} 
\end{theorem}
\begin{proof}
We begin with a CDQS protocol for the function $f$, which by definition achieves $\epsilon, \delta\leq 0.09$ and hides a single qubit secret.
We then amplify using theorem \ref{thm:amplify} to achieve $\epsilon', \delta' \leq 0.09 \cdot 2^{-k}$ and a $k$ qubit secret.
We choose $k$ later to achieve the target $\epsilon_p, \delta_p$ parameters for the two-prover proof. 

The amplified CDQS protocol is defined by Alice and Bob's operations, 
\begin{equation}
    \mathbfcal{N}^x_{QL\rightarrow M_A} \: , \quad
    \mathbfcal{N}^y_{R\rightarrow M_B} \: ,
\end{equation}
as well as a shared resource state, $\Psi_{LR}$. 
In constructing an interactive proof, we will consider unitaries $\mathbf{U}^x_{QAL\rightarrow M_AM_A'}$ and $\mathbf{U}^y_{BR\rightarrow M_BM_B'}$ which purify the above channels, where $A$ and $B$ are any required local ancilla. 
Note that we can always find such purifications with
\begin{equation}\label{eq:systembounds}
    d_{M_A'} \leq d_Qd_Ld_{M_A} \: , \qquad
    d_{M_B'} \leq d_R d_{M_B} \: .
\end{equation}
Further, we distribute a state $\ket{\Psi}_{ELR}$ purifying $\Psi_{LR}$ and with $EL$ held by Alice and $R$ held by Bob.
For convenience we will relabel $EL\rightarrow L$. 
Then, Alice and Bob execute the following: 
\begin{itemize}
    \item Alice sends prover 1 $n_Q$ random bits, $s$. 
    \item Prover 1 sends $M_A$ to Alice and $M_B$ to Bob; prover 2 sends $M_A'$ to Alice and $M_B'$ to Bob. 
    \item Alice executes $(\mathbf{U}^x_{QAL\rightarrow M_AM_A'})^\dagger$; Bob executes $(\mathbf{U}^y_{BR\rightarrow M_BM_B'})^\dagger$.
    \item Bob sends $R$ to Alice. 
    \item Alice measures $A$ in the standard basis; Bob measures $B$ in the standard basis. Additionally, Alice measures $LR$ in a basis that includes $\ket{\Psi}_{LR}$, and measures $Q$ in the standard basis. If the measurements of $A$ and $B$ all return $0$, the measurement of $LR$ returns $\ket{\Psi}_{LR}$, and the measurement of $Q$ returns $s$, Alice outputs accept.
\end{itemize}
The total communication cost of this protocol is
\begin{align}
    \log (d_{M_A}d_{M_A'}d_{M_B}d_{M_B'}) + \log d_R &\leq 2\log d_{M_A} + 2\log d_{M_B} + \log d_Q + \log d_R \nonumber \\
    &\leq 2\,\overline{\text{CDQS}}(f) + k \: ,
\end{align}
where we used the inequalities \eqref{eq:systembounds} from above. 
Thus so long as we can take $k$ to be constant, this protocol has the required cost. 
We proceed to study correctness and security of the protocol.

\vspace{0.2cm}
\noindent \textbf{Correctness:} We consider the case where $(x,y)\in f^{-1}(1)$. 
Observe that running the CDQS protocol forwards would produce a state $\epsilon$ close to $\varphi_{M'P}^{x,y}\otimes \ketbra{s}{s}_Q$ where the referee holds $PQ$ and Alice and Bob hold the purifying system $M'$. 
This state results from the referee applying the recovery operation $\mathbf{V}_{M\rightarrow PQ}^{x,y}$ to the message system he receives, $M$.
In the two-prover proof, we have the two provers prepare $\varphi_{M'P}^{x,y}$ in advance and have prover 1 hold $P$ and prover 2 hold $M'$. 
Then, upon receiving the random string $s$, prover $1$ applies $(\mathbf{V}_{M\rightarrow PQ}^{x,y})^\dagger$ to prepare system $M$, and then each prover then sends $M_A, M_A', M_B, M_B'$ to Alice and Bob according to the pattern shown in figure \ref{fig:CDSandtwoprovers}. 

Now, when Alice and Bob implement $(\mathbf{U}^x_{QAL\rightarrow M_AM_A'})^\dagger\otimes (\mathbf{U}^y_{BR\rightarrow M_BM_B'})^\dagger$ they are completing running the CDQS protocol in reverse, so will produce the inputs to the protocol and pass the test.
In fact, since the protocol is only approximately correct, the output is only approximated by a state with the secret stored on the right, so the test is only passed with some high (but not exactly 1) probability. 
In appendix \ref{sec:QAMlowerbounddetails}, we show that $\epsilon'$-correctness of the CDQS protocol implies $2\sqrt{\epsilon'}$-correctness of the two-prover proof. 
Since $\epsilon'=\epsilon 2^{-k}$, choosing $k$ a large enough constant ensures the two-prover proof is $\epsilon_p$ correct. 

\vspace{0.2cm}
\noindent \textbf{Soundness:} Now suppose that $(x,y)\in f^{-1}(0)$. 
The provers will try to convince Alice to accept. 
Their probability of doing so is
\begin{align}
    p_{\text{pass}} = \frac{1}{d_Q} \sum_s \bra{\psi^s} \rho^s \ket{\psi^s}_{MM'} \: .
\end{align}
Here, $\rho^s_{MM'}$ is the density matrix describing the state prepared by the provers, and $\ket{\psi^s}$ is the state
\begin{align}
    \ket{\psi^s}_{MM'} = (\mathbf{U}^x_{QAL\rightarrow M_aM_a'}\otimes \mathbf{U}^y_{BR\rightarrow M_bM_b'}) \ket{\Psi}_{LR}\ket{00}_{AB} \ket{s}_Q \: .
\end{align}
The states $\rho^{s}_{MM'}$ and $\ket{\psi^s}_{MM'}$ each obey constraints, which will combine to mean the success probability $p_{\text{pass}}$ is small. 

The constraint on $\rho_{MM'}^{s}$ is that, because only prover $P$ (who prepares $M$) is given the random bits, $\rho_{M'}^s=\sigma_{M'}$ is independent of $s$. 

The constraint on $\ket{\psi^s}_{MM'}$ is that the $\psi^s_{M'}$ have nearly orthogonal support. Intuitively, this occurs because of lemma \ref{lemma:complementaryCDS}: since $s$ is not stored in $M$, it is stored in the purifying system $M'$. 
Since $s$ is stored in $M'$ the density matrices on $M'$ must be distinguishable for different values of $s$. 
In appendix \ref{sec:QAMlowerbounddetails} we make this precise, showing that
\begin{align}\label{eq:orthogonalFs}
    \forall s\neq s',\,\,\, F(\psi^s_{M'}, \psi^{s'}_{M'}) \leq 4\sqrt{\delta'} \: .
\end{align}
We then bound the passing probability by
\begin{align}
    p_{\text{pass}} = \frac{1}{d_Q} \sum_s \bra{\psi^s} \rho^s \ket{\psi^s}_{MM'} &= \frac{1}{d_Q} \sum_s F(\psi^s_{MM'}, \rho^s_{MM'}) \nonumber \\
    &\leq \frac{1}{d_Q} \sum_s F(\psi^s_{M'}, \sigma_{M'}) \: .
\end{align}
where in the last line we used that the fidelity increases under the partial trace.
The provers can prepare an arbitrary density matrix on $M'$, so we need to consider the maximum over $\sigma_{M'}$. 
We now use the bound \ref{eq:handyFbound}, which constrains a maximization of this form. 
Using that, we have
\begin{align}
    p_{\text{pass}} &  \leq \max_{\sigma_{M'}} \frac{1}{d_Q} \sum_s F(\psi^s_{M'}, \sigma_{M'}) \nonumber \\
    &\leq \max_{\sigma_{M'}} \frac{1}{d_Q} \sum_s \sqrt{F}(\psi^s_{M'}, \sigma_{M'}) \nonumber \\
    &\leq \frac{1}{d_Q}\sqrt{\sum_{s,s'} \sqrt{F}(\psi^s_{M'}, \psi^{s'}_{M'})} \nonumber \\
    &\leq \frac{1}{d_Q}\sqrt{\sum_{s,s'} (\delta_{ss'}+2(\delta')^{1/4})} \nonumber \\
    &\leq  \frac{1}{d_Q}\sqrt{d_Q+2d_Q^2(\delta')^{1/4}} \nonumber \\
    &= \sqrt{\frac{1}{d_Q}+2(\delta')^{1/4}} \:.
\end{align}
We can recall that $d_Q=2^k$ and $\delta'= \delta 2^{-k}$, so that
\begin{align}
    p_{\text{pass}} \leq \sqrt{\frac{1}{2^k}+\delta 2^{-k/4}} \: .
\end{align}
Since this goes to zero at large $k$, we can choose $k$ to be a large enough constant so that $p_{\text{pass}}\leq \delta_p$.
\end{proof}

\section{Classical-quantum separations}

\subsection{Separating perfectly correct CDS and CDQS}\label{sec:perfectcorrectseparation}

In this section we investigate if quantum resources can provide advantages in implementing the conditional disclosure of secrets primitive. 
In section \ref{sec:perfectcorrectseparation} we show an unconditional separation in the setting of perfectly correct CDS.
Our approach is similar to \cite{kawachi2021communication}, who prove a separation in the perfectly correct and perfectly secure setting for PSM --- we relax the perfect security requirement and adapt this to CDS. 

We will prove a separation for the not-equals function, which recall is defined by
\begin{align}
    \text{NEQ}_n(x,y) = \begin{cases}
        0 \qquad x=y \\
        1 \qquad x\neq y
    \end{cases}
\end{align}
where $x,y$ are $n$ bit inputs. 
We work in a promise setting where either $x=y$ or $x$ and $y$ differ in exactly 
$n/2$ locations. 
To get a separation, we need a lower bound on the classical setting and an upper bound on the quantum setting. 

\vspace{0.2cm}
\noindent \textbf{Classical lower bound:} Our lower bound for NEQ begins with a lower bound on perfectly correct CDS proven in \cite{applebaum2021placing}.
This lower bound is given in terms of the $\coNP^{\cc}$ complexity, which we define next. 

\begin{definition} \textbf{($\coNP^{\cc}$)} A $\coNP^{\cc}$ communication protocol for a function $f:X\times Y\rightarrow \{0,1\}$ is implemented by two parties, which we call Alice and Bob. 
Alice receives input $x\in \{0,1\}^n$ and Bob receives input $y\in \{0,1\}^n$. 
Both Alice and Bob are given a proof $w$.
They then independently decide to accept or reject. 
The $\coNP^{\cc}$ communication complexity of a function $f$, denoted $\coNP^{\cc}(f)$, is the smallest number $c\in \mathbb{N}$ such that:
\begin{itemize}
    \item For any input $(x,y)\in f^{-1}(0)$, there exists a witness $w\in \{0,1\}^c$ such that both Alice and Bob accept when given $w$. 
    \item For any input $(x,y)\in f^{-1}(1)$, there does not exist a witness $w$ such that both Alice and Bob accept when given $w$.
\end{itemize}
\end{definition}

Theorem 3 from \cite{applebaum2021placing} shows that
\begin{align}\label{eq:coNPlowerbound}
    \text{pcCDS}(f) \geq \left(\frac{1}{4}-o(1) \right)\coNP^{\cc}(f) - \log(n) \: ,
\end{align}
where the left hand side denotes the communication cost for perfectly correct CDS.
Considering the randomness cost of the CDS protocol instead, which we will denote $\text{pc}\overline{\text{CDS}}(f)$, we have more simply
\begin{align}
    \text{pc}\overline{\text{CDS}}(f) \geq \frac{1}{2} \coNP^{\cc}(f) \: .
\end{align}
These bounds were proven in the context of unrestricted inputs; an examination of their proof technique however shows that the reduction from perfectly correct CDS to the $\coNP^{\cc}$ setting holds input by input. 
Consequently, the above bounds are also true in the promise setting: the randomness cost to implement perfectly correct CDS for the function $f$ with a given promise on the inputs is lower bounded by the $\coNP^{\cc}$ complexity, considered with the same promise on the inputs. 

Next, we recall that the $\coNP^{\cc}$ complexity is the logarithm of the number of rectangles needed to cover the $0$ entries in the communication matrix \cite{Kushilevitz_Nisan_1996}, without covering any $1$ entries.
We will show that a zero rectangle cannot be too big for the NEQ problem with the given promise on the inputs. 
We use the following theorem from \cite{frankl1987forbidden}, letting
$\Delta(s,t)$ denote the Hamming distance between $s$ and $t$. 
\begin{theorem}\label{thm:0rectangletheorem}
    Let $n$ be divisible by $4$. Let $S,T\subseteq \{0,1\}^n$ be two families of $n$-bit vectors such that for every pair $s\in S, t\in T$ we have $\Delta(s,t)\neq n/2$. Then $|S|\times |T|\leq 2^{2\times 0.96n}$.  
\end{theorem}
This lets us establish the following lower bound. 
\begin{lemma}\label{thm:NEQlowerbound}
    Consider NEQ$_n(x,y)$ with $n$ divisible by $4$ and with the promise that either $x=y$ or $x$ and $y$ differ in exactly half their bits. 
    Then the $\textnormal{pcCDS}$ cost is $\Omega(n)$. 
\end{lemma}
\begin{proof}
Suppose $S\times T$ is a 0 rectangle. 
This means that for $s\in S$, $t\in T$, we have $\Delta(s,t) \neq n/2$. 
Then theorem \ref{thm:0rectangletheorem} above says that $|S|\times |T|\leq 2^{2\times 0.96n}$, so the size (area) of any $0$ rectangle is at most $2^{2\times 0.96n}$. 

On the other hand, suppose we can use less than $n/100$ bits of communication. 
Then letting $N$ be the number of rectangles in the resulting covering of the 0 entries, we have
\begin{align}
    \log N = \frac{n}{100} \: .
\end{align}
Let $s_{\text{max}}$ be the size of the largest $0$ rectangle. 
A rectangle of size $s$ can only cover $\sqrt{s}$ of the 0's on the diagonal, which means that the number of rectangles must be at least
\begin{align}
    N \geq \frac{2^{n}}{\sqrt{s_{\text{max}}}} \: ,
\end{align}
so then
\begin{align}
    \log \left(\frac{2^{n}}{\sqrt{s_{\text{max}}}}\right) \leq \frac{n}{100} \: .
\end{align}
Solving this for $s_{\text{max}}$, we find
\begin{align}
    2^{2\times 0.99 n} \leq s_{\text{max}} \: .
\end{align}
But earlier we found that $s_{\text{max}}$ must be smaller than the above, which is a contradiction. 
Thus, any such protocol must use more than $n/100$ bits of communication.
\end{proof}

\vspace{0.2cm}
\noindent \textbf{Quantum upper bound:} We next proceed to give a logarithmic upper bound using quantum resources. 
We exploit a solution from \cite{buhrman2010nonlocality} to the following distributed Deutsch-Jozsa problem. 
Alice and Bob receive inputs $x\in\{0,1\}^n$ and $y\in\{0,1\}^n$, with the above mentioned promise. 
Their goal will be to produce a pair of shorter strings $a,b$ which are equal if and only if $x=y$. 
The idea for our CDS protocol is to use $a,b$ as (shorter) inputs to a CDS protocol for NEQ. 

\begin{lemma}\label{thm:quantumupperbound}
    Consider NEQ$_n(x,y)$ with $n$ a power of $2$ and with the promise that either $x=y$ or $x$ and $y$ differ in exactly half their bits. 
    Then the $\textnormal{pc}\overline{\textnormal{CDQS}}$ cost, including both communication and shared entanglement, is $O(\log n)$.
\end{lemma}

\begin{proof}
Alice and Bob share $\log n$ EPR pairs, and both apply controlled phase gates to prepare the state
\begin{align}
    \ket{\Psi_1} = \frac{1}{\sqrt{n}}\sum_{i} (-1)^{x_i+y_i} \ket{i}_A\ket{i}_B.
\end{align}
Next they both apply the Hadamard operation to obtain
\begin{align}
    \frac{1}{n\sqrt{n}} \sum_{a,b,i} (-1)^{x_i+y_i+i\cdot a + i\cdot b} \ket{a}_A\ket{b}_B .
\end{align}
Alice and Bob now both measure in the computational basis. 
The probability of obtaining any pair of outcomes $(a,b)$ such that $a=b$ is then 
\begin{align}
    \left|\frac{1}{n\sqrt{n}} \sum_{i}(-1)^{x_i+y_i} \right|^2.
\end{align}
This is $1/n$ when $x=y$, and (because of the promise) 0 otherwise. 
Thus Alice and Bob obtain strings $a, b$ of length $\log n$ which are always equal when $x=y$, and never equal otherwise. 

Now, Alice and Bob run a classical CDS protocol for the NEQ$_{\log(n)}$ function, with $a,b$ as their inputs. 
This can be implemented with linear (in $\log n$) communication and randomness \cite{applebaum2021placing} in the perfect setting, so they use $\log n$ communication and randomness plus the $\log n$ EPR pairs we used to shorten the inputs. 
\end{proof}

Lemmas \ref{thm:NEQlowerbound} and \ref{thm:quantumupperbound} together imply a quantum-classical separation for CDS in the case of perfect correctness. 
We can also notice that since the quantum upper bound does not introduce any soundness errors, the same observations separate pCDS and pCDQS, the versions of CDS and CDQS with both perfect correctness and security.

\subsection{Exponential separation of PSQM and PSM for a partial function}\label{sec:PSMseparations}

In this section we revisit the topic of separations for PSM in the robust (imperfect security and correctness) setting. 
In \cite{kawachi2021communication} the authors point out that there is a relational problem with an exponential classical-quantum separation in the robust case, and a separation for a partial function in the exact setting. 
Here we show that the exponential separation can be achieved for a partial function in the robust setting.

To show our separation, we use the following version of the Boolean Hidden Matching problem defined by Kerenidis and Raz~\cite{kerenidisraz}. 
The problem uses the notion of a perfect matching, which is an ordered list of $n/2$ pairs $(i,j)$, $i,j\in [n]$ such that each $i\in[n]$ occurs exactly once in the matching.  

\begin{definition}\label{def:BHM} The \textbf{Boolean Hidden Matching} (BHM) problem is defined by:
\begin{itemize}
    \item \textbf{Inputs:} Alice receives $x\in \{0,1\}^{2n}$ and Bob receives an ordered perfect matching $M$ on $[2n]$ and a string $w\in \{0,1\}^n$.
    \item \textbf{Output:} 1 if $Mx+w$ has Hamming weight at least $2n/3$, 0 if $Mx+w$ has Hamming weight less than $n/3$
\end{itemize}
We are promised that one of the output conditions is true.
$Mx$ refers to the $n$-bit string whose $k^{\text{th}}$ component is $x_i+x_j$, where $(i,j)\in M$ is the $k^{\text{th}}$ pair in the matching (in order). All operations are performed over $\mathbb{F}_2$. 
\end{definition}

\paragraph*{Classical Lower Bound:} BHM was previously used to give an exponential separation between one-way quantum and one-way randomized communication complexity \cite{GKKRW07}. In particular, theorem 4 in~\cite{kerenidisraz} implies that BHM requires $\Omega(\sqrt{n})$ communication in the classical one-way model, and hence in the classical simultaneous model as well. Finally, we observe that any PSM protocol also gives a classical simultaneous protocol and hence, any PSM for BHM requires $\Omega(\sqrt{n})$ communication. 

\paragraph*{Quantum Upper Bound.} The intuition for the PSQM for BHM is that, when Alice and Bob share entanglement, using local measurements there is a randomized reduction from the BHM problem to the inner product problem on $O(\log n)$ bits. We describe this in more detail below.

\begin{theorem}
    Considering the BHM problem given in definition \ref{def:BHM}, there is a protocol that computes this problem in the PSQM model using $O(\log n)$ shared EPR pairs and $O(\log n)$ classical communication.
\end{theorem}
\begin{proof}\,
Alice and Bob start off with $\log(2n)$ EPR pairs 
\begin{align}
\frac{1}{\sqrt{2n}}\sum_{i\in[2n]}\ket{i}_A\ket{i}_B \: .
\end{align}
Alice adds her input $x$ in the phase to produce
\begin{align}
    \frac{1}{\sqrt{2n}}\sum_{i\in[2n]}  (-1)^{x_i} \ket{i}_A\ket{i}_B \: .
\end{align}
Bob measures with $n$ projectors $E_{i,j} \equiv \ket{i}\bra{i}+\ket{j}\bra{j}$ for each edge $(i,j)\in M$ in his perfect matching. 
This gives him a random $(i,j)\in M$ and the state ignoring the normalization is 
\begin{align}
     (-1)^{x_i}\ket{i}_A\ket{i}_B + (-1)^{x_j}\ket{j}_A\ket{j}_B \: .
\end{align}
Now, Alice and Bob each apply the Hadamard gate to all their qubits to obtain 
\begin{align}
    \sum_{k,l\in[2n]}\left[(-1)^{\langle k+l,i\rangle + x_i} + (-1)^{\langle k+l,j\rangle +  x_j})\right]\ket{k}_A\ket{l}_B \: ,
\end{align}
where $\langle \cdot ,\cdot \rangle$ denotes the inner product over $\mathbb{F}_2$. The players then measure all their registers. 
Observe that $\ket{k}\ket{l}$ has non-zero amplitude if and only if $(-1)^{\langle k+l,i\rangle + x_i}=(-1)^{\langle k+l,j\rangle + x_j}$, in other words,
\begin{align}
    \langle k+l,i+j\rangle= x_i+x_j \: .
\end{align}
Recall the promise on the input that either $x_i+x_j+w_{i,j}=0$ for at least $2n/3$ many $(i,j)\in M$, or $x_i+x_j+w_{i,j}=1$ for at least $2n/3$ many $(i,j)\in M$.
From the above, we can replace $x_i+x_j$ by $\langle k+l,i+j\rangle$. 
Thus, the problem reduces to computing $\langle k+l,i+j\rangle + w_{i,j}$ and testing whether it is mostly 0 or mostly 1 for a uniformly random $(i,j)\sim M$. 
We now embed this into a single instance of the inner product function. 
Recall that $k$ belongs to Alice and $l,i,j,w$ belong to Bob. 
Thus, Bob can compute $b'=\langle l,i+j\rangle$, $b=i+j$ and the players need to compute $\langle k,b\rangle +b'+w_{i,j}$ which can be viewed as the inner product between $(k,1,1)$ and $(b,b',w_{i,j})$ where Alice knows $k\in [2n]$ and Bob knows $b\in[2n],b',w_{i,j}\in\{0,1\}$. 

Altogether, when Alice and Bob share entanglement, by performing local measurements, they can do a randomized reduction to an instance of the inner-product function on $O(\log n)$ bits (where the randomness is over Bob's measurement outcome $(i,j)\sim M$). Since this reduction only involves local measurements and no communication, the referee learns nothing. Lemma 3 from~\cite{cryptoeprint:2018/144} implies that the inner product function on $O(\log n)$ bits has a PSM of cost $O(\log n)$. Alice and Bob then execute this PSM for the inner product function on the inputs $(k,1,1)$ and $(b,b',w_{i,j})$ respectively and compute the desired inner product $\langle k,b\rangle + b'+w_{i,j}$ using $O(\log n)$ communication. The security of this protocol follows from the security of the PSM for inner-product~\cite{cryptoeprint:2018/144}.
\end{proof}

\section{An upper bound for forrelation}

In this section we give a logarithmic upper bound on CDQS for the forrelation problem, which we define more precisely below. 
The strategy combines techniques from non-local quantum computation (NLQC) and communication complexity. 

We first make some comments on non-local quantum computation. 
A non-local quantum computation is any process realized in the form shown in figure \ref{fig:non-localcomputation}. 
Typically, the goal of an NLQC is to implement a joint channel $\mathbfcal{N}_{AB}$ on two quantum systems $A$, $B$, with $A$ initially held by Alice and $B$ initially held by Bob, as shown in figure \ref{fig:local}. 
Alice and Bob each act locally on their systems (plus their portions of a shared entangled state), exchange one simultaneous round of quantum or classical communication, then act locally again. 
The overall transformation realized in this process should be (or should approximate) $\mathbfcal{N}_{AB}$. 

\begin{figure*}
    \centering
    \begin{subfigure}{0.45\textwidth}
    \centering
    \begin{tikzpicture}[scale=0.6]
    
    \draw[thick] (-1,-1) -- (-1,1) -- (1,1) -- (1,-1) -- (-1,-1);
    
    \draw[thick] (-3.5,-3) to [out=90,in=-90] (-0.5,-1);
    \draw[thick] (3.5,-3) to [out=90,in=-90] (0.5,-1);
    
    \draw[thick] (0.5,1) to [out=90,in=-90] (3.5,3);
    \draw[thick] (-0.5,1) to [out=90,in=-90] (-3.5,3);
    
    \node at (0,0) {$\mathbfcal{N}_{AB}$};
    
    \node at (0,-5) {$ $};
    
    \end{tikzpicture}
    \caption{}
    \label{fig:local}
    \end{subfigure}
    \hfill
    \begin{subfigure}{0.45\textwidth}
    \centering
    \begin{tikzpicture}[scale=0.4]
    
    \draw[thick] (-5,-5) -- (-5,-3) -- (-3,-3) -- (-3,-5) -- (-5,-5);
    \node at (-4,-4) {$\mathbfcal{V}^L$};
    
    \draw[thick] (5,-5) -- (5,-3) -- (3,-3) -- (3,-5) -- (5,-5);
    \node at (4,-4) {$\mathbfcal{V}^R$};
    
    \draw[thick] (5,5) -- (5,3) -- (3,3) -- (3,5) -- (5,5);
    \node at (4,4) {$\mathbfcal{W}^R$};
    
    \draw[thick] (-5,5) -- (-5,3) -- (-3,3) -- (-3,5) -- (-5,5);
    \node at (-4,4) {$\mathbfcal{W}^L$};
    
    \draw[thick] (-4.5,-3) -- (-4.5,3);
    
    \draw[thick] (4.5,-3) -- (4.5,3);
    
    \draw[thick] (-3.5,-3) to [out=90,in=-90] (3.5,3);
    
    \draw[thick] (3.5,-3) to [out=90,in=-90] (-3.5,3);
    
    \draw[thick] (-3.5,-5) to [out=-90,in=-90] (3.5,-5);
    \draw[black] plot [mark=*, mark size=3] coordinates{(0,-7.05)};
    
    \draw[thick] (-4.5,-6) -- (-4.5,-5);
    \draw[thick] (4.5,-6) -- (4.5,-5);
    
    \draw[thick] (4.5,5) -- (4.5,6);
    \draw[thick] (-4.5,5) -- (-4.5,6);
    
    \end{tikzpicture}
    \caption{}
    \label{fig:non-localcomputation}
    \end{subfigure}
    \caption{(a) Circuit diagram showing the local implementation of a channel $\mathbfcal{N}_{AB}$. (b) Circuit diagram showing the non-local implementation of the same channel. The operations $\mathbfcal{V}^L$, $\mathbfcal{V}^R$, $\mathbfcal{W}^L$, and $\mathbfcal{W}^R$ are quantum channels. The lower, bent wire represents an entangled state. }
    \label{fig:non-localandlocal}
\end{figure*}
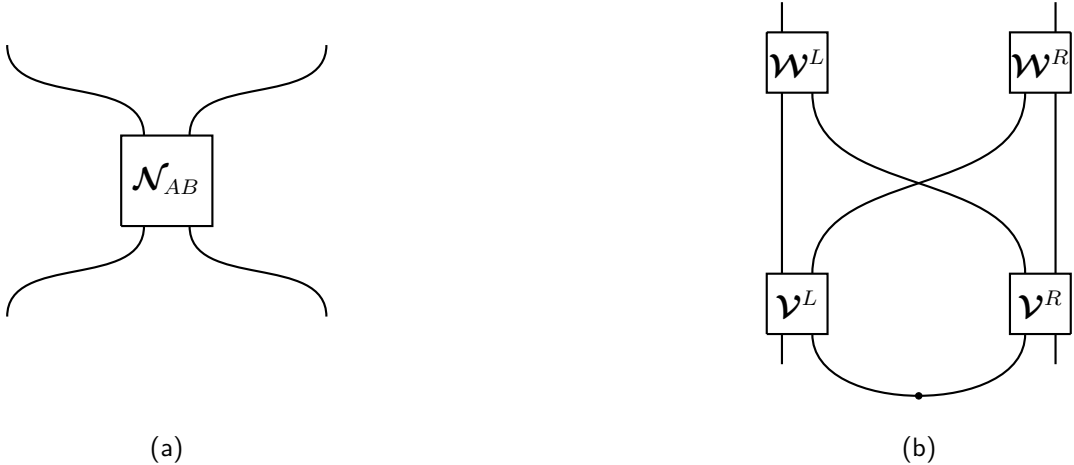

The key result from non-local quantum computation we will make use of is the following, reproduced from \cite{speelman2015instantaneous}. 
\begin{theorem}\label{thm:Tdepth}
    Any n-qubit quantum circuit $C_{AB}$ using the Clifford+$T$ gate set which has $T$-depth $d$ has a protocol for instantaneous non-local computation using $O((68n)^d)$ EPR pairs.
\end{theorem}
A further comment is that the communication used in the protocol that realizes this upper bound has the same scaling as the entanglement cost. 

Next, we define the forrelation problem, for which we will give a CDQS upper bound using this NLQC technique. Let $n$ be a power of 2. 
Define the forrelation of a string $x\in\{-1,1\}^{n}$ to be
\begin{align}\label{eq:def_forr}
    \forr(x) := \frac{1}{n}\bra{x_1}H^{\otimes n}\ket{x_2}
\end{align}
where $x_1$ is the vector formed by first $n/2$ bits of $x$, and $x_2$ is vector formed by the final $n/2$ bits of $x$. This problem was defined by~\cite{Aar10,AA15,raztal} in the context of oracle separations between quantum and classical query complexity. Following this, \cite{girish2022quantum} defined a communication complexity version of this problem as follows.

\begin{definition}\label{def:forrelation}
Alice is given input $x\in\{-1,1\}^n$ and Bob is given $y\in \{-1,1\}^n$, with $n$ a power of 2. 
Then, to solve the \textbf{Forrelation} problem Alice should output the value $f(x,y)$ defined by
\begin{align}
    f(x,y) = \begin{cases}
        -1 \quad \text{if}\,\, \forr(x\cdot y)\geq \alpha \\
        +1 \quad \text{if}\,\, \forr(x\cdot y)\leq \beta
    \end{cases}
\end{align}
with $\alpha>\beta>0$ and $\alpha-\beta$ constant. Here, $x\cdot y$ denotes the point-wise product of $x$ and $y$.
\end{definition}

A variant of this problem with a $1/\log n$ gap was studied by~\cite{girish2022quantum} who showed that it has a quantum simultaneous protocol of $O(\log^3 n)$ cost when Alice and Bob share $O(\log^3 n)$ EPR pairs, but the classical randomized communication cost is $\tilde{\Omega}(n^{1/4})$. In this paper, we consider a constant gap. The advantage of this is that since the gap is constant, the quantum protocol only needs to be amplified a constant number of times, which turns out to be important later. On the other hand, it can be shown by combining~\cite{girish2022quantum} and~\cite{BS21}
that the classical hardness persists even with constant gap. 
 
In part, our interest in this problem is as a candidate problem to be outside of $\AM^{\cc}$. Indeed, the query problem of estimating $\forr(x)$ as in equation \ref{eq:def_forr} is known to be outside of $\AM$ in the query complexity world~\cite{raztal,BS21}. If the communication version of the Forrelation problem is indeed outside of $\AM^{\cc}$, then (as explained in the introduction) our upper bounds for this problem would give a quantum-classical separation for robust CDS. 

\begin{figure}
\centering 
	\mbox{ 
		\Qcircuit @C=1em @R=.7em {
			& \ket{0} &  & \gate{H} &  \multigate{7}{E} &  \qw & \multigate{3}{\mathcal{O_A}} & \qw & \multigate{7}{E}  & \qw & \targ &\ctrl{1}  &\ctrl{2} &\qw & \ldots& & \ctrl{3} & \gate{H} & \meter \qw  \\ 
			& \ket{0}  &  & \gate{H}& \ghost{E} & \qw  & \ghost{\mathcal{O_A}} & \qw  & \ghost{E} & \qw &\qw &  \gate{H} & \qw & \qw & \ldots& &\qw &\qw &    \\
			& \ket{0}  &  & \gate{H}&  \ghost{E} &\qw  & \ghost{\mathcal{O_A}} & \qw  & \ghost{E} & \qw &  \qw &\qw &  \gate{H} &\qw & \ldots& &\qw &\qw &    \\
			& \ket{0}  & & \gate{H} & \ghost{E} & \qw & \ghost{\mathcal{O_A}} & \qw & \ghost{E} & \qw  &  \qw &\qw &\qw &  \qw&  \ldots& &\gate{H} &\qw \\
			& \ket{0} &  & \qw& \ghost{E} &\qw & \multigate{3}{\mathcal{O_B}} & \qw & \ghost{E} & \qw \\ 
			& \ket{0}  &  &  \qw& \ghost{E} &\qw & \ghost{\mathcal{O_B}} & \qw  & \ghost{E} & \qw  \\
			& \ket{0} &  &  \qw& \ghost{E} &\qw & \ghost{\mathcal{O_B}} & \qw  & \ghost{E} & \qw  \\
			& \ket{0} &   &\qw  & \ghost{E} &\qw &  \ghost{\mathcal{O_A}} & \qw & \ghost{E} & \qw \\
		}
	}
\caption{Circuit computing the forrelation function $f(x,y)$. Figure reproduced from \cite{girish2022quantum}. Here, $E$ is the operator as in~\ref{fig:E}.}
\label{fig:forrcircuit}
\end{figure}
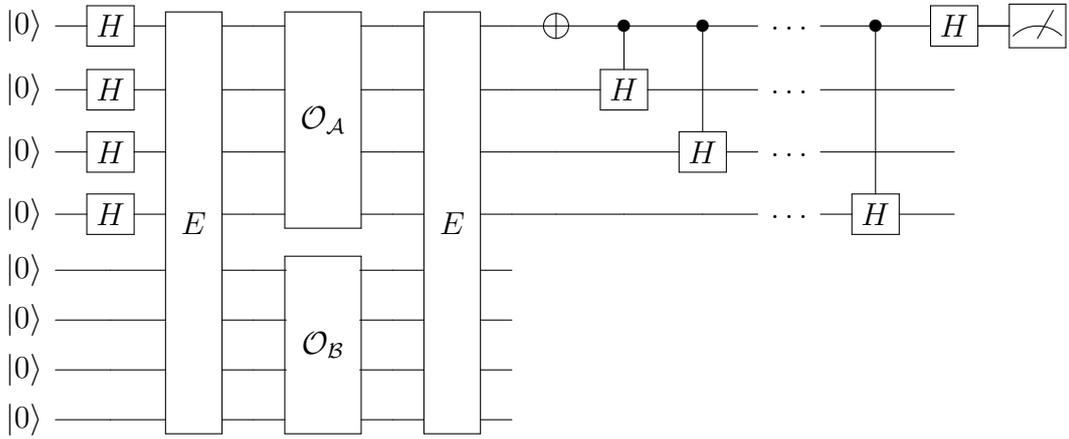

\begin{figure}[h!]
\centering
\mbox{ 
\Qcircuit @C=1em @R=.7em {
& \ctrl{3} & \qw & \qw & \ldots && \qw & \qw\\
& \qw & \ctrl{3} & \qw & \ldots && \qw &\qw \\
& \qw & \qw & \qw & \ldots&& \ctrl{3} &\qw\\
&  \targ & \qw & \qw & \ldots && \qw &\qw\\
&  \qw & \targ& \qw & \ldots&& \qw  &\qw \\
&    \qw & \qw& \qw & \ldots && \targ  &\qw\\
}
}
\caption{The $E$ operator}
\label{fig:E}
\end{figure}
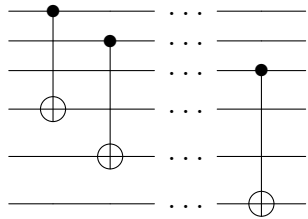

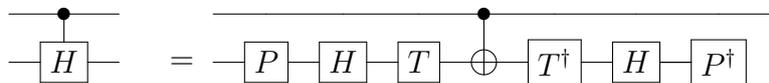
\begin{figure}[h!]
    \centering
    \mbox{ 
\Qcircuit @C=1em @R=.7em {
&\ctrl{1} &  \qw  & &   &  & \qw & \qw & \qw &\ctrl{1} & \qw  &\qw & \qw & \qw\\
&\gate{H} & \qw &  & = & & \gate{P}  & \gate{H} & \gate{T} &  \targ & \gate{T^\dagger} & \gate{H} & \gate{P^\dagger} \qw \\
}}
    \caption{Convenient decomposition of a controlled Hadamard operator as in figure 6 in~\cite{controlledhadamard}. Notice this can be expressed without any $T$ gates on the upper wire. This allows us to express the circuit of figure \ref{fig:forrcircuit} in constant $T$-depth. Figure reproduced from \cite{girish2022quantum}.}
    \label{fig:CH}
\end{figure}

We give our upper bound next. 
\begin{theorem}
    There exists a constant $c>0$ such that the forrelation problem given in definition \ref{def:forrelation} can be implemented in the $\textnormal{CDQS}$ setting with $O(\log^cn)$ communication and $O(\log^cn)$ randomness. 
\end{theorem}

\begin{proof}\,
The key observation is that there is a way to compute forrelation, labelled $f(x,y)$, in constant $T$-depth. 
We then leverage the connection between NLQC and CDQS, and the upper bounds in terms of $T$-depth for NLQC. 

From here we just need to observe that there is a circuit that computes $f(x,y)$ in constant $T$-depth. 
For this, we make use of the circuit given in \cite{girish2022quantum}, shown in figure \ref{fig:forrcircuit}. 
The circuit uses $\log n$ qubits. 
Further, we make use of the decomposition of the controlled $H$ gate shown in figure \ref{fig:CH}.
To implement this circuit non-locally, we view the first layer of $H$ gates, $E$, and the oracle calls as preparing the input state $\ket{\psi_{xy}}$. 
In other words, Alice and Bob share $O(\log (2n))$ EPR pairs and first introduce a phase into this shared state based their individual inputs to produce $\ket{\psi_{xy}}$. 
Then we take $U_{AB}$ to implement the remaining portions of the circuit. 
$E$ is Clifford, and $X$ is Clifford. 
Using the decomposition of the controlled $H$ gate in figure \ref{fig:CH} we can implement the cascading controlled H operators using just two layers of $T$ gates. 

The remaining steps to compute $f(x,y)$ are then to perform the measurement and amplify the outcome by repeating the circuit. 
We note that the measurement returns one bit, and to amplify we need to repeat only $O(1)$ times (since the gap $\alpha-\beta$ is constant) and take the majority. 
The circuit implementing majority acts on $O(1)$ qubits, so contributes at most $O(1)$ to the $T$-depth.
Thus the entire circuit is constant $T$-depth, and theorem \ref{thm:Tdepth} gives a polynomial (in the circuit size) upper bound in terms of both entanglement and communication.  
Since the circuit here has size $\log n$ for $n$ the number of input bits, the entire protocol is implemented with $\text{poly}(\log n)$ communication and entanglement. 
\end{proof}

\section{Discussion}

In this work we explored the differences and analogies between quantum and classical CDS. 
We've done so with a few goals in mind: to better understand the power of quantum resources in information-theoretic cryptography, to better understand classical CDS, and to better understand non-local quantum computation, of which quantum CDS can be understood as a special case. 

We established that indeed quantum resources can provide advantages for CDS by finding a separation for perfectly correct CDS. 
We also gave a novel protocol for forrelation, which suggests an advantage in the robust case as well. 
Exploring the analogies between lower bounds for classical and quantum CDS, we proved a lower bound on quantum CDS from $\QAM[2,2]^{\cc}$, a two-prover, two-message interactive proof setting. 
This is, so far, the closest analogue bound to the classical bound from $\AM^{\cc}$. 

One property of classical CDS for which we could find no quantum analogue is randomness sparsification \cite{applebaum2021placing}, which gives that the randomness cost of a CDS protocol never needs to be larger than the communication cost, up to possible logarithmic differences. 
In the quantum case we were unable to determine if this is also true. 
This seems closely related to the analogous problem in quantum communication complexity, where it is also unknown if entanglement larger than the communication can ever be helpful. 

A key open question in the study of classical CDS is to establish linear lower bounds for explicit functions in the robust setting, or to better understand obstructions to doing so.\footnote{See \cite{applebaum2023advisor} for some discussion around understanding obstructions.} 
One motivation for studying the quantum case is to bring a new perspective and set of tools to bear on this problem. 
Indeed, for the perfectly secure setting the quantum perspective provides a new lower bound \cite{asadi2024rank}. 
Our lower bound on robust CDQS in terms of one-way classical communication complexity reveals a weakness in this bound as applied to the classical case: somehow the bound does not see enough of the structure of a CDS protocol to distinguish between quantum and classical protocols. 
We hope further exploration of quantum lower bounds will yield insight into the difficult problem of finding good lower bounds on classical CDS.  

\subsection*{Acknowledgements}

We thank Henry Yuen and Tal Malkin for useful discussions.
AM and CW acknowledge the support of the Natural Sciences and Engineering Research Council of Canada (NSERC); this work was supported by an NSERC Discovery grant (RGPIN-2025-03966) and NSERC-UKRI Alliance grant (ALLRP 597823-24). UG is supported by an NSF award (CCF-232993) and LO is supported by an NSF Graduate Fellowship.
Research at the Perimeter Institute is supported by the Government of Canada through the Department of Innovation, Science and Industry Canada and by the Province of Ontario through the Ministry of Colleges and Universities. 

\appendix

\section{Proof details for the lower bound from \texorpdfstring{$\QAM[2,2]^{\cc}$}{TEXT}}\label{sec:QAMlowerbounddetails}

We provide details for the proof of theorem \ref{thm:QAMbound}.

\vspace{0.2cm}
\noindent \textbf{Details on correctness:} We give the calculation showing that $\epsilon$-correctness of the CDQS protocol gives $2\sqrt{\epsilon}$ correctness of the two-prover proof. 
Let 
\begin{align}
    \mathbfcal{V}_{M\rightarrow PQ}(\cdot) &= \mathbf{V}_{M\rightarrow PQ}(\cdot) \mathbf{V}_{M\rightarrow PQ}^\dagger \nonumber \\
    \mathbfcal{U}_{QAL\rightarrow M_AM_A'}(\cdot) &= \mathbf{U}_{QAL\rightarrow M_AM_A'}(\cdot) \mathbf{U}_{QAL\rightarrow M_AM_A'}^\dagger \nonumber \\
    \mathbfcal{U}_{BR\rightarrow M_BM_B'}(\cdot) &= \mathbf{U}_{BR\rightarrow M_BM_B'}(\cdot) \mathbf{U}_{BR\rightarrow M_BM_B'}^\dagger.
\end{align}
Note that from correctness of the CDQS protocol we have that there exists, for all $(x,y)\in f^{-1}(1)$, a channel $\mathbfcal{D}^{x,y}_{M\rightarrow Q}$ such that
\begin{align}\label{eq:diamondcorrectness}
    \left\Vert\mathbfcal{D}^{x,y}_{M\rightarrow Q}\circ \mathbfcal{N}^{x,y}_{Q\rightarrow M} - \mathbfcal{I}_{Q\rightarrow Q}\right\Vert_\diamond \leq \epsilon.
\end{align}
Using equation \eqref{eq:opunderdiamond}, we can also obtain that there exists an isometric extension of these channels which is close in operator norm. 
Since one isometric extension of $\mathbfcal{D}^{x,y}_{M\rightarrow Q}\circ \mathbfcal{N}^{x,y}_{Q\rightarrow M}$ is $\mathbfcal{V}_{M\rightarrow PQ}\circ (\mathbfcal{U}^x_{QAL\rightarrow M_AM_A'}\otimes \mathbfcal{U}^y_{BR\rightarrow M_BM_B'})\circ \Psi_{\varnothing\rightarrow LR}$ (where $\Psi_{\varnothing\rightarrow LR}$ prepares the state $\ket{\Psi}_{LR}$) and all isometric extensions are related by an isometry on the purifying system, we have that all isometric extensions of $\mathbfcal{D}^{x,y}_{M\rightarrow Q}\circ \mathbfcal{N}^{x,y}_{Q\rightarrow M}$ can be expressed in the form
\begin{align}
    \mathbfcal{S}_{PM'} \circ \mathbfcal{V}_{M\rightarrow PQ}\circ (\mathbfcal{U}^x_{QAL\rightarrow M_AM_A'}\otimes \mathbfcal{U}^y_{BR\rightarrow M_BM_B'})\circ \Psi_{\varnothing\rightarrow LR}
\end{align}
where $\mathbfcal{S}_{PM'}(\cdot) = \mathbf{S}_{PM'}(\cdot)\mathbf{S}^\dagger_{PM'}$ with $\mathbf{S}_{PM'}$ an isometry.
Further, isometric extensions of the identity channel must just append a state preparation, 
\begin{align}
    \mathbfcal{I}_Q &\rightarrow \mathbfcal{I}_Q\otimes \mathbfcal{W}_{\varnothing \rightarrow PM'}.
\end{align}
Now, we employ equation \eqref{eq:isometricextensionbound} to bound the diamond norm between these isometric extensions in terms of the diamond norm between the channels, which itself is bounded by $\epsilon$ from equation \ref{eq:diamondcorrectness}, obtaining
\begin{align}
&\inf_{\mathbfcal{S}, \mathbfcal{W}} \left\Vert\mathbfcal{S}_{PM'} \circ \mathbfcal{V}_{M\rightarrow PQ}\circ (\mathbfcal{U}^x_{QAL\rightarrow M_AM_A'}\otimes \mathbfcal{U}^y_{BR\rightarrow M_BM_B'})\circ \mathbf{\Psi}_{\varnothing\rightarrow ELR} -  \mathbfcal{I}_Q\otimes \mathbfcal{W}_{\varnothing \rightarrow PM'}\right\Vert_\diamond \\ &\leq 2\sqrt{\epsilon}.\nonumber 
\end{align}
Using isometric invariance of the diamond norm, we can rewrite this as
\begin{align}
    \inf_{\mathbfcal{W}} \left\Vert\mathbfcal{V}_{M\rightarrow PQ}\circ (\mathbfcal{U}^x_{QAL\rightarrow M_AM_A'}\otimes \mathbfcal{U}^y_{BR\rightarrow M_BM_B'})\circ \mathbf{\Psi}_{\varnothing\rightarrow LR} -  \mathbfcal{I}_Q\otimes \mathbfcal{W}_{\varnothing \rightarrow PM'}\right\Vert_\diamond \leq 2\sqrt{\epsilon}
\end{align}
and further as
\begin{align}
    \inf_{\mathbfcal{W}} \left\Vert\mathbfcal{I}_{QAB} \otimes \mathbf{\Psi}_{\varnothing\rightarrow LR} -  (\mathbfcal{U}^x_{QAL\rightarrow M_AM_A'}\otimes \mathbfcal{U}^y_{BR\rightarrow M_BM_B'})^\dagger \circ \mathbfcal{V}_{M\rightarrow PQ}^\dagger \circ \mathbfcal{W}_{\varnothing \rightarrow PM'}\right\Vert_\diamond \leq 2\sqrt{\epsilon}.
\end{align}
We have the provers begin with the state $\ket{\varphi^{x,y}}_{PM'}$ that is output by the optimizing $\mathbfcal{W}_{\varnothing \rightarrow PM'}$.
Now, we use the definition of the diamond norm and consider the input state $\ket{s}_Q \ket{00}_{AB}$ to find that
\begin{align}
    \left\Vert\ket{s}_Q\ket{00}_{AB}\ket{\Psi}_{LR} - (\mathbf{U}^x_{QAL\rightarrow M_AM_A'}\otimes \mathbf{U}^y_{BR\rightarrow M_BM_B'})^\dagger \circ \mathbf{V}_{M\rightarrow PQ}^\dagger\ket{\varphi^{x,y}}_{PM'} \right\Vert_1 \leq 2\sqrt{\epsilon}.
\end{align}
In terms of the fidelity this is, 
\begin{align} &F(\ket{s}_Q\ket{00}_{AB}\ket{\Psi}_{LR}, (\mathbf{U}^x_{QAL\rightarrow M_AM_A'}\otimes \mathbf{U}^y_{BR\rightarrow M_BM_B'})^\dagger \circ \mathbf{V}_{M\rightarrow PQ}^\dagger\ket{\varphi^{x,y}}_{PM'}\ket{s}_Q)\\
    &\geq 1-2\sqrt{\epsilon}
\end{align}
but also 
\begin{align}
    p_{accept} = |\bra{s}_Q\bra{00}_{AB}\bra{\Psi}_{LR} (\mathbf{U}^x_{QAL\rightarrow M_AM_A'}\otimes \mathbf{U}^y_{BR\rightarrow M_BM_B'})^\dagger \circ \mathbf{V}_{M\rightarrow PQ}^\dagger\ket{\varphi^{x,y}}_{PM'}\ket{s}_Q|^2
\end{align}
so that the probability of Alice accepting is at least $1-2\sqrt{\epsilon}$ for any choice of secret $s$, and hence also at least this when averaged over $s$. 
Choosing $k \geq \log (\epsilon/\epsilon_p)$ (a constant), we can amplify the protocol sufficiently to achieve the needed correctness parameter.

\vspace{0.2cm}
\noindent \textbf{Soundness:} Here we show equation \eqref{eq:orthogonalFs}, which expresses that the reduced density matrices $\psi^s_M$ are nearly orthogonal for distinct $s$. 
To make this precise, begin with lemma \ref{lemma:complementaryCDS} which gives that there exists $\mathbfcal{D}^{x,y}_{M'\rightarrow Q}$ such that
\begin{align}
    \left\Vert\mathbfcal{D}^{x,y}_{M'\rightarrow Q} \circ (\mathbfcal{N}^{x,y})^c_{Q\rightarrow M'}-\mathbfcal{I}_Q \right\Vert_\diamond \leq 2\sqrt{\delta}.
\end{align}
Acting on the input $\ketbra{s}{s}$, this gives
\begin{align}
    \left\Vert\mathbfcal{D}^{x,y}_{M'\rightarrow Q} \circ (\mathbfcal{N}^{x,y})^c_{Q\rightarrow M'}(\ketbra{s}{s}_Q)-\ketbra{s}{s}_Q \right\Vert_1 \leq 2\sqrt{\delta}.
\end{align}
Now consider
\begin{align}
    \left\Vert \mathbfcal{D}^{x,y}_{M'\rightarrow Q} \circ (\mathbfcal{N}^{x,y})^c_{Q\rightarrow M'}(\ketbra{s}{s}_Q)-\mathbfcal{D}^{x,y}_{M'\rightarrow Q} \circ (\mathbfcal{N}^{x,y})^c_{Q\rightarrow M'}(\ketbra{s'}{s'}_Q) \right\Vert_1.
\end{align}
Inserting $\ketbra{s}{s}-\ketbra{s}{s}+\ketbra{s'}{s'}-\ketbra{s'}{s'}$ and applying the reverse triangle inequality and triangle inequality, we obtain
\begin{align}
    \left\Vert\mathbfcal{D}^{x,y}_{M'\rightarrow Q} \circ (\mathbfcal{N}^{x,y})^c_{Q\rightarrow M'}(\ketbra{s}{s}_Q)-\mathbfcal{D}^{x,y}_{M'\rightarrow Q} \circ (\mathbfcal{N}^{x,y})^c_{Q\rightarrow M'}(\ketbra{s'}{s'}_Q) \right\Vert_1 &\geq \left\Vert\ketbra{s}{s} - \ketbra{s'}{s'}\right\Vert_1 - 2\sqrt{\delta} \nonumber \\
    &= 2(1-2\sqrt{\delta}). \nonumber 
\end{align}
But also, by monotonicity of the trace distance, 
\begin{align}
    \left\Vert\rho_{M'}^s- \rho_{M'}^{s'} \right\Vert_1 &= \left\Vert(\mathbfcal{N}^{x,y})^c_{Q\rightarrow M'}(\ketbra{s}{s}_Q)- (\mathbfcal{N}^{x,y})^c_{Q\rightarrow M'}(\ketbra{s'}{s'}_Q) \right\Vert_1 \nonumber \\
    &\geq \left\Vert\mathbfcal{D}^{x,y}_{M'\rightarrow Q} \circ (\mathbfcal{N}^{x,y})^c_{Q\rightarrow M'}(\ketbra{s}{s}_Q)-\mathbfcal{D}^{x,y}_{M'\rightarrow Q} \circ (\mathbfcal{N}^{x,y})^c_{Q\rightarrow M'}(\ketbra{s'}{s'}_Q) \right\Vert_1 \nonumber 
\end{align}
so that we obtain
\begin{align}
    \frac{1}{2}\left\Vert\rho_{M'}^s- \rho_{M'}^{s'} \right\Vert_1 \geq 1-2\sqrt{\delta}. 
\end{align}
Translating this to a bound on the fidelity via the Fuch's van de Graff inequality, we obtain
\begin{align}
    \forall s\neq s',\,\,\, F(\psi^s_{M'}, \psi^{s'}_{M'}) \leq 4\sqrt{\delta}
\end{align}
as needed. 

\bibliographystyle{unsrtnat}
\bibliography{biblio}

\end{document}